\documentclass[twoside]{article}

\usepackage[accepted]{aistats2025}
\usepackage{natbib}

\usepackage{amsmath,amsfonts,bm}

\def\eqref#1{equation~\ref{#1}}

\def\1{\bm{1}}

\DeclareMathAlphabet{\mathsfit}{\encodingdefault}{\sfdefault}{m}{sl}
\SetMathAlphabet{\mathsfit}{bold}{\encodingdefault}{\sfdefault}{bx}{n}

\newcommand{\R}{\mathbb{R}}

\DeclareMathOperator*{\argmax}{arg\,max}
\DeclareMathOperator*{\argmin}{arg\,min}

\usepackage{hyperref}
\usepackage{url}
\usepackage{booktabs} 
\usepackage{graphicx}

\usepackage[utf8]{inputenc}

\usepackage{amssymb}

\usepackage{dsfont}

\usepackage{algorithmicx}
\usepackage[noend]{algpseudocode}
\algnewcommand\algorithmicforeach{\textbf{for each}} 
\algdef{S}[FOR]{ForEach}[1]{\algorithmicforeach\ #1\ \algorithmicdo} 

\usepackage{optidef} 

\usepackage[ruled,vlined]{algorithm2e}

\newtheorem{theorem}{Theorem}
\newtheorem{theorem*}{Theorem}
\newtheorem{definition}{Definition}

\newtheorem{proof}{Proof}
\newtheorem{proof*}{Proof}

\newcommand{\indicator}{\bm{1}}

\usepackage{tikz}
\usetikzlibrary{arrows,}

\usepackage{graphicx}
\graphicspath{ {images/} }
\usepackage[font={small,it}]{caption} 
\usepackage{float}
\usepackage{subcaption}

\usepackage{bm, bbm} 

\newcommand{\TP}{\textsf{TP}}
\newcommand{\FP}{\textsf{FP}}

\usepackage{calrsfs}
\DeclareMathAlphabet{\pazocal}{OMS}{zplm}{m}{n}

\newcommand{\X}{\pazocal{X}}
\newcommand{\Y}{\pazocal{Y}}
\newcommand{\XX}{\textbf{X}}
\newcommand{\YY}{Y}

\newcommand{\ZZ}{\textbf{Z}}
\renewcommand{\SS}{\textbf{S}}

\newcommand{\x}{\mathbf{x}}
\newcommand{\z}{\mathbf{z}}

\newcommand{\manip}{\text{manip}}
\newcommand{\rec}{\text{rec}}

\usepackage{comment}

\usepackage[capitalize]{cleveref} 

\begin{document}

\runningauthor{Yatong Chen, Andrew Estornell, Yevgeniy Vorobeychik, Yang Liu}

%

%

\twocolumn[

\aistatstitle{To Give or Not to Give? The Impacts of Strategically Withheld Recourse}

\aistatsauthor{ 
  Yatong Chen \And Andrew Estornell 
}
\aistatsaddress{ 
  MPI for Intelligent Systems,\\ T\"{u}bingen AI Center, T\"{u}bingen, Germany \And  Bytedance Research
}

\aistatsauthor{ 
  Yevgeniy Vorobeychik \And Yang Liu 
}

\aistatsaddress{ 
  Washington University in Saint Louis \And University of California, Santa Cruz 
}

]

\begin{abstract}

  Individuals often aim to reverse undesired outcomes in interactions with automated systems, like loan denials, by either implementing system-recommended actions (recourse), or manipulating their features.
  While providing recourse benefits users and enhances system utility, it also provides information about the decision process that can be used for more effective strategic manipulation, 
   especially when the individuals collectively share such information with each other. 
   We show that this tension leads rational utility-maximizing systems to frequently withhold recourse, resulting in decreased population utility, particularly impacting sensitive groups.
   To mitigate these effects, we explore the role of recourse subsidies, finding them effective in increasing the provision of recourse actions by rational systems, as well as lowering the potential social cost and mitigating unfairness caused by recourse withholding. 
\end{abstract}

\section{INTRODUCTION}
When individuals interacting with automated systems are denied a desired outcome (e.g., loan approval), they may seek a means of reversing this decision to obtain the desired outcome. 
This procedure is commonly referred to as \emph{recourse} \citep{ustun2019actionable}. In cases where the system's decision rule is opaque (e.g., lending), the system itself is responsible for supplying individuals with recourse, i.e., a recommended feature modification that is feasible and will result in that individual being approved. 

When the feature modification changes an agent's true qualification rate (e.g., paying off debt increases one's creditworthiness), providing recourse can benefit the system. However, offering recourse actions also exposes information about the system's decision rule, as each action leads to a positively classified feature vector close to the decision boundary. This added transparency creates opportunities for strategic individuals to exploit the system's decision rule by manipulating their features, especially when they share their knowledge about the decision rule with one another. For example, platforms like \emph{GradCafe} for graduate school admissions and \emph{LendingClub} for loan applications allow agents to see other applicants' features. This enables them to potentially \emph{misreport} their features to mimic those of others, thereby leveraging publicly available information to their advantage \citep{bechavod2022information,chen2020linear,estornell2021unfairness,hardt2016strategic,vorobeychik2023many}. 
Such feature manipulation can often reduce both system and social utility since it will increase the false positive rate. 
This creates a tension in providing recourse, where the utility gained from increased qualifications must be balanced against the utility lost due to manipulation that exploits the counterfactual information in recourse recommendations
The consequence of this tension is that in many settings, providing recourse to all, or even most, of the agents may be suboptimal from a system's perspective.
This sharply contrasts with the common assumption in the algorithmic recourse literature, which typically considers agents taking recourse actions without the possibility of manipulation. 

On the other hand, we can consider subsidies as a means to incentivize systems to offer recourse actions. Subsidy \citep{hu2019disparate}, or government incentive, is a type of government expenditure to financially help individuals, households, and businesses in various settings. Consider the \emph{small business administration (SBA) microloan program} \footnote{https://www.hud.gov/program-offices/housing/fhahistory} in the United States as a motivating example. This program provides small loans to startups and small businesses and offers technical assistance and financial training to help borrowers succeed. In this work, we model subsidies that lower each individual's recourse costs, requiring the agent to pay only a fraction of the original amount. Compared with using penalty to disincentivize manipulation \citep{blocki2013audit} or using auditing to incentivize recourse taking \citep{estornell2023incentivizing}, the main benefit of subsidies is that it requires no verification power from the system, reducing the potential harm caused by unintentionally impose large fines on truthful agents. We add a more detailed discussion in \cref{sec:related-work}.

\paragraph{High-Level Overview of Our Model}
There are two parties in our setting: a utility-maximizing recourse system and a set of agents. Each agent is represented by a feature vector $\mathbf{x} \in \X$. The system trained a \emph{fixed}, potentially opaque function $f: \X \rightarrow [0, 1]$ to decide who to provide a resource (e.g., loan) based on $\x$. For negatively classified agents, the system decides whether to provide a recourse action or not. The central tension comes from the fact that agents can both (1) lie about their features and manipulate them to some publicly known positively classified features and (2) take the recommended recourse actions that change their true features. Only the latter leads to an increase in the system's utility. 
The publicly known features come from either agents who are already classified positively, or agents who successfully obtain a recourse action from the system. The latter is more within the system's control and could be an easier target for manipulation, as they are more likely to be closer to the decision boundary. Thus, the system's main tool is to strategically withhold recourse actions from some agents to maximize their utility. Based on the relative cost of recourse and manipulation, agents choose to take the recommended action or manipulate known positively classified features. See \Cref{fig:system-agent-interaction} for a demonstration of our modeling framework.

\paragraph{Main Results}
We show that in many cases, the system is incentivized to strategically withhold recourse from most if not all, agents to prevent manipulation. To our knowledge, this is the first work to challenge the assumption that a utility-maximizing recourse system will naturally provide recourse without third-party intervention (e.g., government regulation). As fewer agents receive recourse, the \emph{social cost}—the average cost to achieve positive classification—rises. Withholding recourse also limits legitimate paths to positive classification, pushing more individuals toward manipulation. This burden often falls disproportionately on disadvantaged groups, worsening existing inequalities. To address this, we explore recourse subsidies, a third-party payment that reduces recourse costs, and find them effective in increasing recourse providing, reducing social costs, and mitigating unfairness.

The details for reproducing our experimental results can be found at {\small \texttt{\url{https://github.com/UCSC-REAL/Strategic-withheld-recourse}}}.

\section{RELATED WORKS}
\label{sec:related-work}
Our work is closely related to the literature on algorithmic recourse, strategic classification, and fairness in general. Due to the page limit, additional related work on fairness and social cost in strategic classification and recourse \citep{gupta2019equalizing, vonkügelgen2022fairness, ehyaei2023robustness, estornell2021unfairness}, transparency \citep{barsotti2022transparency, akyol2016price} and others can be found in \Cref{sec:additional-related-work}.

\paragraph{Recourse} 
Much of the line of algorithmic recourse \citep{ustun2019actionable, venkatasubramanian2020philosophical,karimi2020survey,gupta2019equalizing,karimi2020algorithmic, vonkugelgen2020fairness,chen2020linear, harris2022bayesian} focuses on the setting where the requested recourse is guaranteed to be provided out of ethical consideration \citep{venkatasubramanian2020philosophical}.
Our work is the first to challenge this fundamental assumption and argue that without a third-party's intervention, a utility-maximizing algorithmic recourse system may be incentivized to withhold recourse from some agents to prevent manipulations strategically. We point the reader to \cite{karimi2020survey} for a more detailed discussion of the concepts and recent development of algorithmic recourse. 

\paragraph{Strategic Classification} Strategic classification focuses on the problem of how to effectively make predictions in the presence of agents who behave strategically to obtain desirable outcomes \citep{hardt2016strategic,chen2018strategyproof,tsirtsis2019optimal,levanon2021strategic,dong2018strategic,chen2018strategyproof, zrnic2021leads}. In this work, we use the standard game-theoretic Stackelberg model proposed in \cite{hardt2016strategic} to simulate the agent’s best response actions when choosing between recourse and manipulation. Our work considers the \emph{imitation-based} manipulations: 
agents do not know the classifier $f$ but are aware of a set of positively classified features and can misreport their feature by imitating another agent's feature that is positively classified. Such copycat behavior has been well-known in the literature of game theory, the behavioral economy, and strategic classification, e.g., \citep{bechavod2022information, barsotti2022transparency}. 
While most of this line of work focuses on agents being strategic and could potentially modify their features to get a favorable prediction outcome, our work focuses on when the system is being strategic and potentially withholds recourse to the agents. 

\paragraph{Subsidy, Penalty, and Auditing} Our work relates to interventions aimed at (dis)incentivizing strategic behaviors. Most relevant is \cite{hu2019disparate}, who also studies strategic behavior using subsidies. Penalties for misreporting \citep{hardt2016strategic, blocki2013audit} offer another way to discourage manipulation, encouraging agents to pursue recourse instead. Both subsidies and penalties can be viewed as tools to shift the balance between the cost of recourse and manipulation — penalties raise the cost of manipulation, while subsidies lower the cost of recourse. \citet{estornell2023incentivizing} explores auditing as an intervention to promote recourse, assuming universal recourse availability. The implementation of penalties requires verification power, such as in tax systems where cross-checking reported income deters misreporting. Subsidies, however, could be financed by third-party entities like governments or financial institutions. Incentivizing recourse through penalties is not ideal, as verification can lead to false positives, unfairly penalizing truthful agents. Audit-based systems \citep{estornell2023incentivizing} typically impose large fines, which can harm innocent agents if they are wrongly identified as manipulators. Subsidies avoid this issue. If the system controls audits and subsidies alone, it will prioritize its utility, which may not always align with the population's best interests.

\section{PRELIMINARIES}
\label{sec:preliminaries}

Let $\X\subset \mathbb{R}^d$ and $\Y\equiv\{0, 1\}$ be a domain of features and labels respectively. Let $f:\X\rightarrow \Y$ be a fixed binary classifier.
A population of agents with features $\XX = \{x: x \in\X\}$ and labels $Y = \{y: y\in \Y\}$ are classified by $f$, which is unknown to the agents; all agents desired to be positively classified (e.g., all loan applicants desire approval).
Denote the domain of negatively classified features as $\X_{-}\subseteq \mathbb{R}^d$ and the domain of positively classified features as $\X_{+}\in \mathbb{R}^d$, i.e. $f(\x) = 0$ for all $\x\in\X_{-}$ and $f(\x)=1$ for all $\x\in\X_{+}.$
All agents prefer positive classification over negative classification. Agents who have features $\x\in\X_{-}$ have two means of obtaining positive classification in the next step, \emph{recourse} and \emph{manipulation}, which are defined next.
\paragraph{Recourse} 
Recourse provides agents who received undesirable outcomes with recommended actions to genuinely improve their outcome by modifying their attributes \citep{ustun2019actionable}. 
Let $c_R: \X \times \X \rightarrow \R_{+}$ be the cost of recourse, i.e. an agent with true features $\x$ pays cost $c_R(\x, \x')$ when modifying their features to be $\x'$. 
An agent with true feature $\x\in\X_{-}$ has an \emph{optimal} recourse action
\footnote{Throughout the paper, we will interchangeably use the terms 'recourse action' and 'recourse feature.' They both refer to the feature vector that will be classified positively after the agent's taking a particular recourse action. In other words, we assume that whenever an agent reveals their recourse \emph{action}, it also reveals their original feature vector, which is equivalent to revealing the feature vector that corresponds to the vector\emph{after} the agent performs recourse.},
\begin{align}\label{eq:opt_rec}
    \x_R(\x)  = \operatorname{argmin}_{\x'\in \X_+}c_R(\x, \x')\\
    \text{s.t.}~f(\x') = 1,~~ \x'\in A(\x)  \nonumber
\end{align}
where $A(\x)$ represents the set of features an agent with true features $\x$ can feasibly obtain, i.e., the \emph{actionable} recourse actions provided by the system.
When agents perform recourse, both their true features and true qualification rate change, i.e., their true features become $\x_R(\x)$, and their true qualification rate changes from $\Pr(y = 1|\x)$ to $\Pr(y = 1|\x_R(\x))$.

\paragraph{Manipulation} In addition to recourse, agents can also perform manipulations. Following \cite{barsotti2022transparency}, we focus on \emph{imitation-based} manipulations: 
agents do not know the classifier $f$, but are aware of a set of publically revealed positively classified features $\ZZ \subseteq \X_+$ (defined below) and can misreport their feature by imitating another agent’s feature that is positively classified and is publically revealed. For a manipulation cost function $c_M: \X \times \X \rightarrow \mathbb{R}_{+}$ the \emph{optimal} imitation-based manipulation for an agent with true feature $\x$ is
\begin{align}\label{eq:opt_manip}
    \x_M (\x) = \arg\min_{\x'\in \ZZ}c_M(\x, \x')
\end{align}

Different from recourse, manipulation is simply a misreport rather than a change of one’s features, thus it does not change $\Pr[y = 1|\x]$. However, since the system only observes the reported features before classification, it does not know whether a report is truthful.

\paragraph{Feature Disclosure and Publicaly Revealed Set $\ZZ$}
We model the set of publicly revealed features $\ZZ\subseteq \X_+$ resulting from agents sharing information with each other.
In particular, $\ZZ$ 
consists of features that may come from two sets: 1) the revealed recourse actions recommended by the system (i.e., $\z\in \XX_R$ where $\XX_R =\{\x_R(\x), \x\in \XX_{-}\}$),
and 2) the set of initial positively classified features (i.e., $\z\in \XX_{+}$). Each element is made public with a \emph{fixed} probability $p\in [0, 1]$, and all publicly revealed elements make the reveal set $\ZZ$. 
We represent the set of recourse actions that are actually revealed as $\ZZ_R = \{\z \in \XX_R : \text{Reveal}(\z) = 1\}$. Here, $\text{Reveal}(\z)$ is a random indicator function that equals 1 with probability $p$ (indicating that feature $\z$ is revealed) and 0 otherwise. Similarly, let $\ZZ_+$  represent the positively classified features that are actually revealed: $\ZZ_+ = \{\z \in \XX_+ : \text{Reveal}(\z) = 1\}$. As a result, $\ZZ = \ZZ_R \cup \ZZ_+$.  

This captures real-life scenarios where negatively classified agents collectively gather information about classifier $f$ by observing positively classified peers or those who obtained recourse. Revealed recourse features are particularly crucial as they lie near the decision boundary, making them more likely targets for manipulation than general positive features.

\section{INTERACTION BETWEEN AGENTS AND THE SYSTEM}
\label{sec:interaction}
\begin{figure}
    \centering
    \includegraphics[width=\linewidth]{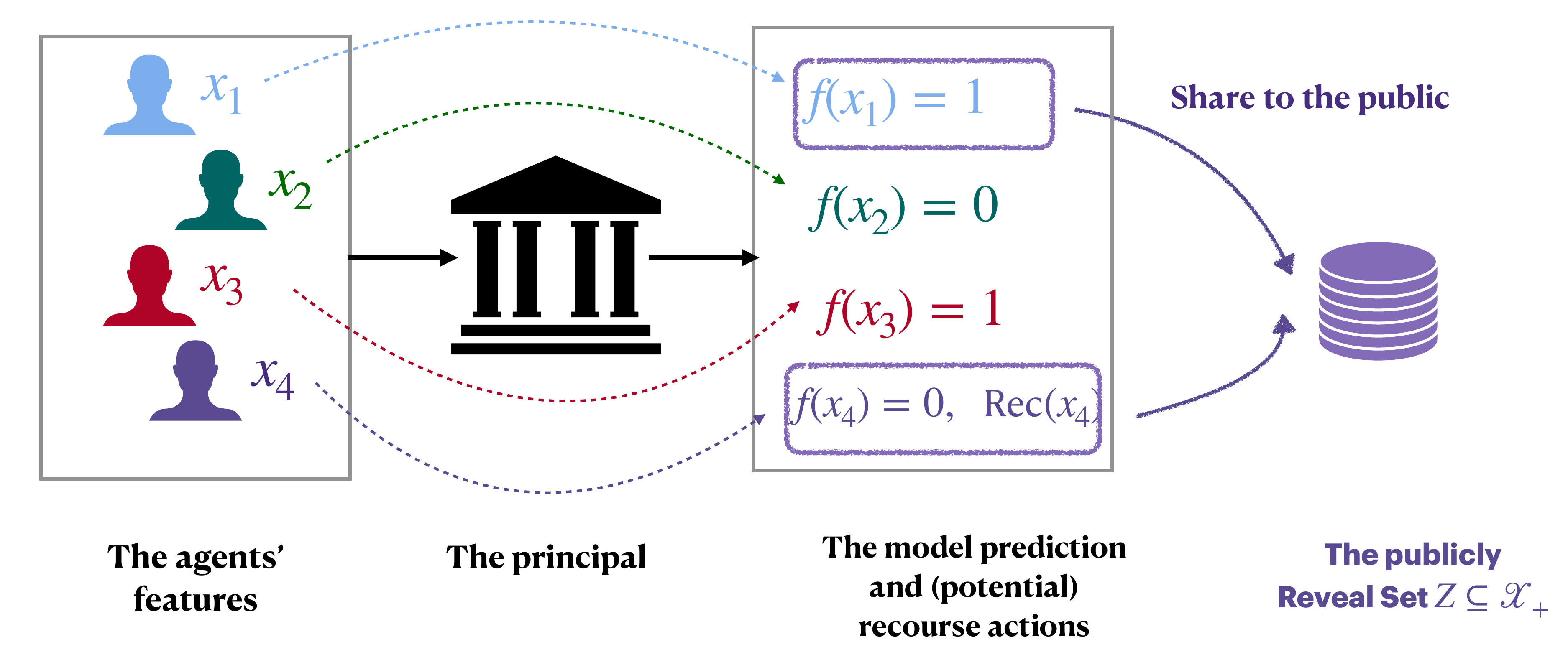}
    \caption{\emph{Demonstration of our modeling framework. Agents arrive simultaneously, and the system trains a classifier $f: \X \rightarrow \Y$ for maximum prediction accuracy. Negatively classified agents request recourse, and the system selects agents for recourse provision to maximize utility (\Cref{eq:sys_obj}). Positively classified agents and those provided recourse have a probability $p \in [0, 1]$ to reveal features, contributing to the publicly revealed set $\ZZ \subseteq \X_+$. Upon observing $\ZZ$, agents execute final actions based on \Cref{eqn:final-action}.}}
    \label{fig:system-agent-interaction}
\end{figure}
Unlike the traditional recourse setting, where the system is expected to provide recourse to any individual upon request, without external regulation (e.g., government mandates requiring banks to offer recourse), a utility-maximizing system may have incentives to withhold recourse to prevent strategic manipulation by agents. In this section, we introduce our modeling framework to capture these dynamics.

\paragraph{A Motivating Example}  
\emph{A bank publishes a classifier to determine who qualifies for a credit card. Each applicant (with feature vector $\x$) is approved if the bank’s model $f$ predicts they can repay their loan. For applicants denied a card, the bank may offer recourse, i.e., a plan to improve their creditworthiness, such as paying off debt or increasing their income. These recourse actions are provided through specific programs, such as financial classes.}
\emph{Agents also have access to an online forum where some applicants share their approved loan or recourse features. With knowledge of both recourse actions and the forum, some agents may misreport their features to match positively classified ones in an attempt to gain approval without actually taking the recommended recourse actions. As a result, the bank may have an incentive to limit recourse to individuals whose features are harder to manipulate (e.g., features that are easier for the bank to verify).}

We now formalize the dynamics between the recourse system and the agents. 

\textbf{System:} The system trains a classifier $f: \X \rightarrow \Y$ to maximize the prediction accuracy: 
\[
    f = \argmax_{f\in \cal{F}} \sum_{x\in \XX} \mathbbm{1} [f(\x) = y]
\]
A collection of negatively classified agents with features $\XX_{-} \subseteq \X_{-}$ will request recourse actions from the system after receiving their prediction outcome. The system first computes optimal recourse actions for all negatively classified agents but only chooses to \emph{release} a subset of those recourse actions $\ZZ_R \subset \XX_{R}$ to the public to maximize its utility, i.e., $\text{TP} - \text{FP}$: 

\begin{align}
\label{eq:sys_obj}
\small
\max_{\ZZ_R\subset \XX_R}&   ~~~\underbrace{\TP({\SS}) - \FP(\SS)}_{\text{system's utility}}\\
 \text{s.t.}&~~ \underbrace{\SS = \{\z(\x, \ZZ): \x\in \XX\}}_{\substack{\text{agent's reported features (Eq \ref{eqn:final-action})}}} \\ &\underbrace{\ZZ = \ZZ_R \cup \ZZ_{+}}_{\substack{\text{all publicly revealed features}}} \nonumber
\end{align}
Here, $\TP(\SS)$ and $\FP(\SS)$ are the true positive and false positive rates on the set of features after the agent's final actions.
We assume that the system either knows $c_R$ and $c_M$, or can reasonably approximate these cost functions when optimizing their objective. Intuitively, this definition of system utility reflects a bank gaining a utility of $1$ for each repaid loan and $-1$ for each defaulted loan.

\textbf{Agents:} Agents who are negatively classified will request a recourse action from the system. Upon seeing the publically revealed features $\ZZ$ defined in \cref{sec:preliminaries}, agents who are provided with a recourse action adapt their features from $\x$ to $\z = \x_M(\x)$ or $\z = \x_R(\x)$ such that $f(\z) = 1$, while minimizing the cost of the corresponding action. When both the recourse and manipulation actions are greater than $1$\footnote{The strategic agent's utility for adapting their feature from $x$ to $x'$ is determined by the standard utility function in the literature of strategic classification (see, e.g., \cite{hardt2016strategic}), which is $U(x, x') = f(x') - c(x, x')$. Thus, when the cost of adaptation $c(x, x')\geq 1$, the utility will be less than 0, in which case, the agent does nothing.}, the agents will choose to stay with their original features $\x$, which corresponds to the \emph{do-nothing} action. Agents who are not provided with a recourse action will choose to manipulate or \emph{do nothing}. The final action for already positively classified agents is always the \emph{do-nothing} action.

\textbf{Agent's best response:} Denote $\zeta_\x\in \{0, 1\}$ as an indicator for whether agent $\x$ is provided with a recourse or not (i.e., $\zeta(\x) = 1$ when provided with a recourse action). Then for all agents with $f(\x)=0$, their final action is:
\begin{align}
\scriptsize
\label{eqn:final-action}
\z(\x, \ZZ) = 
\begin{cases}
    \x_R(\x) & \zeta_\x=1 \text{ and } c_R(\x, \x_R(\x)) < \min(1, c_M(\x, \x')), \\
             & \forall \x' \in \ZZ \\
    \x_M(\x) & \zeta_\x=1 \text{ and } c_M(\x, \x_M(\x)) < \min(1, c_R(\x, \x')), \\
             & \forall \x' \in \ZZ, \text{ or } \zeta_\x=0 \text{ and }~c_M(\x, \x_M(\x)) < 1 \\
    \x       & \zeta_\x=1  \text{ and } c_R(\x, \x_R(\x)), c_M(\x, \x_R(\x)) \geq 1, \\
             & \forall \x' \in \ZZ, \text{ or } \zeta_\x=0 \text{ and } c_M(\x, \x_M(\x)) \geq 1 
\end{cases}
\end{align}
\textbf{Summary of System-Agent Interaction}:
\begin{enumerate}
\item Agents arrive simultaneously, and the system trains a classifier $f: \X \rightarrow \Y$ for maximum prediction accuracy.
\item Negatively classified agents request recourse, and the system selects agents for recourse provision to maximize utility (\Cref{eq:sys_obj}).
\item Positively classified agents and those provided recourse have a probability $p \in [0, 1]$ to reveal features, contributing to the publicly revealed set $\ZZ \subseteq \X_+$.
\item Upon observing $\ZZ$, agents execute final actions based on \Cref{eqn:final-action}.
\end{enumerate}
Our framework is intended to capture settings where black box models are used for decision-making. Any agent subjected to the decision rules will not have direct access to the model but will still act in their best interest. In these opaque settings, recourse proposed by the system naturally offers a way for agents to learn more about the decision rule, thus increasing their ability to game the system.

The following two definitions introduce key metrics that will be used throughout this paper -- the recourse rate quantifies the proportion of negatively classified agents who opt to take recourse actions when presented with a disclosed feature set. The manipulation rate captures the fraction of negatively classified agents who choose to manipulate their features under the same conditions:

\begin{definition}
\label{def:recourse-ratio}
    (Recourse Rate) Let $\XX_{-}$ be the set of features of  negatively classified agents.
    For a given set of disclosed features (i.e., recourse actions) $\ZZ$, the \emph{recourse rate} $\rec(\ZZ, \XX_{-})$ is defined as the fraction of agents who choose to perform recourse when shown $\ZZ$:
    \begin{align*}
    \resizebox{.95\linewidth}{!}{$
        \rec(\ZZ, \XX_{-}) = \frac{\underset{\x\in\XX_{-}}{\sum} \mathbbm{1}
        \big[
        \underset{\z'\in\ZZ}{\min}~c_R(\x, \z') < \min\big(1,~\underset{\z''\in\ZZ}{\min}~c_M(\x, \z'')\big) 
        \big]}{|\mathbf{X}_{-}|}$}
    \end{align*}
\end{definition}

\begin{definition}\label{def:manip-rate}
    (Manipulation Rate) Let $\XX_{-}$ be the set of features of the negatively classified agents.
    For a given set of disclosed features (i.e., recourse actions) $\ZZ$, the \emph{manipulation rate} $\manip(\ZZ, \XX_{-})$ is defined as the fraction of the $n$ agents which choose to manipulate when shown features $\ZZ$:
    \begin{align*}
    \resizebox{.95\linewidth}{!}{$
    \manip(\ZZ, \XX_{-}) = 
        \frac{\underset{\x\in \XX_{-}}{\sum} \mathbbm{1} \big[\underset{\z'\in\ZZ}{\min}~c_M(\x, \z') < \min\big(1,~\underset{\z''\in\ZZ}{\min}~c_R(\x, \z'')\big) \big]}{|\mathbf{X}_{-}|}$}
    \end{align*}
\end{definition}

\section{SYSTEM UTILITY}
\label{sec:system}
Recall from the previous section, the system aims to select a set ${\ZZ}_R\subseteq \XX_{R}$ to reveal as recourse recommendations simultaneously to maximize its utility (\Cref{eq:sys_obj}). 
We can first show that this problem is NP-hard (\cref{thm:nonseq_hard} in \cref{sec:np-hard}).
Despite the hardness of this objective, the system's utility is \emph{submodular} in the set of provided recourse actions (\cref{thm:submodular-objective-function} in \cref{sec:submodular}). This characteristic enables the system to employ standard submodular optimization techniques to approximately get the optimal recourse actions to disclose to $k$ agents.

We can show that in expectation, the system benefits from agents \emph{taking} recourse actions: 

\begin{theorem}(System's Expected Utility Changes)
\label{thm:expected-utility-change}
    The system's expected utility (defined in \cref{eq:sys_obj}) increases for each recourse action \emph{taken} by agents and decreases for every manipulation action \emph{taken} by agents.
    When the classifier used by the system is better than random guessing, which means that $f(x) = 1$ implies $\Pr[y(x) = 1| X = x]\geq 0.5$, then the system's utility is monotonically increasing in each recourse action taken by an agent in expectation but will be monotonically decreasing in each manipulation action taken by an agent.
\end{theorem}
However, this does not imply that the system is always incentivized to provide as many recourse actions as possible, since agents might not always take them if they collude, which creates a natural misalignment between the system's utility and recourse offering for the system. 

\section{COST OF STRATEGICALLY WITHHOLDING RECOURSE SYSTEM}
\label{sec:cost}
Having shown that the system could be incentivized to withhold recourse from the agents,
we now study the consequence of such withholdings by examining the social cost and unfairness as a result of the system's strategic actions.

\begin{definition}
\label{def:social-cost}
(Social Cost of a Strategically Withhold Recourse) Given a publically revealed set $\ZZ \subseteq \X_+$, 
    the \emph{social cost} refers to the additional cost agents must pay as a result of the system withholding recourse.
    Denote $\x_R(\x)$ as the optimal recourse action provided by a non-strategic system, and $\z_R(\x, \ZZ)$ as the recourse action that the agent takes given the revealed set $\ZZ$, then the social cost of a strategically withholding recourse system is defined as:
        \begin{align*}
        \small
          & \text{cost}(\ZZ, \XX_{-}) 
          = \sum_{\x\in \XX_{-}}\big(c_R(\x, \z_R(\x, \ZZ)) - c_R(\x, \x_R(\x))\big)
        \end{align*}
\end{definition}
where $\z_R(\x, \ZZ) = \argmin_{\z\in \ZZ} c_R(\x, \z)$. For the remainder of our results, we focus on univariate classifiers, i.e., the feature $\x$ is one-dimensional.
There is a natural correspondence between univariate and multivariate classifiers in the sense that one can imagine the space of single-dimensional features as the scores produced a multi-dimensional classifier $f(\x)$ \footnote{This follows similarly to Lemma 3.1 in \cite{milli2019social}.}.
That is, in the case when $f(\x) = [h(\x) \geq \theta]$ for some score function $h$ and threshold theta, we can view $f$ as a single dimensional classifier acting on the space of scores produced by $h$. 

We also measure the disparities of different social groups in terms of their differences in 1) recourse ratios (defined in \Cref{def:recourse-ratio}), and 2) social cost (defined in \Cref{def:social-cost}). Understanding the disparities in terms of recourse rate and social cost among different groups is crucial for addressing issues of unfairness in an algorithmic recourse system \cite{gupta2019equalizing, vonkügelgen2022fairness}. These disparities often reflect systemic biases and inequalities, impacting marginalized communities disproportionately. 
In particular, assume there are two groups of agents $\XX^{(g_0)}$ and $\XX^{(g_1)}$, where $g_0, g_1$ represents their group memberships, we are interested in the following quantities:

\begin{definition}(Disparity in Social Cost and Recourse Ratio) The disparity in social cost and recourse ratio for two groups $g_0, g_1$ are defined as:  
    \begin{align*}
    \resizebox{.91\linewidth}{!}{$
    \textit{Diff}^{(\text{cost})}(\ZZ, \XX^{(g_0)}, \XX^{(g_1)}):= \left|\text{cost}(\ZZ, \XX_{-}^{(g_1)}) - \text{cost}(\ZZ, \XX_{-}^{(g_0)})\right|$},\\
    \resizebox{.91\linewidth}{!}{$\textit{Diff}^{(\rec)}(\ZZ, \XX^{(g_0)}, \XX^{(g_1)}):= \left|\rec(\ZZ, \XX_{-}^{(g_1)}) - \rec(\ZZ, \XX_{-}^{(g_0)})\right|$}
    \end{align*}   
\end{definition}

In the experiments section, we demonstrate that these disparities can be quite common across different datasets (see \Cref{fig:lr-rec-social-cost-diff}). By quantifying and illuminating these disparities, we gain crucial insights into the specific mechanisms of inequity and injustice within algorithmic recourse systems.

\section{THE EFFECT OF SUBSIDIES}
\label{sec:subsidy}
To remedy the adverse population- and group-level impacts previously observed, we investigate the use of subsidies (rigorously defined next) and their impact on recourse rate, social cost, and unfairness we defined in the previous section.  
Subsidies correspond to a global decrease in the cost of recourse. 
For example, free educational material on financial literacy distributed to any agent petitioning the bank for recourse will increase the ease at which that agent can perform recourse actions. 
\begin{definition}(Subsidies)\citep{hu2019disparate} A subsidy $0 \leq \alpha \leq 1$ is a scalar decrease to the cost of recourse. For subsidy $\alpha$, agents performing recourse pay only $(1 - \alpha)\cdot c_R(\x, \x')$ instead of the full cost of $c_R(\x, \x') $. We denote $c_R(\x, \x'; \alpha) = (1 - \alpha)\cdot c_R(\x, \x')$ as the new recourse cost at subsidy level $\alpha$. 
\end{definition}
Next, we demonstrate how subsidies can help increase the recourse rate (Theorem \ref{thm:sub-rec-rate}) and system's utility (Theorem \ref{thm:sub-utility}).  Additionally, subsidies can mitigate disparities in recourse rate differences (Theorem \ref{thm:sub-recourse-rate-diff}) and social cost differences (Theorem \ref{thm:sub-social-cost-diff}) among various groups.

We first show how subsidies influence the recourse rate. 
Recall that subsidy reduces the cost of recourse from $c_R(\x, \x')$ to $c_R(\x, \x';\alpha)$. 
With that, the recourse rate becomes:
\begin{align*}
    &\rec(\ZZ, \XX_{-}; \alpha) \\
    =& \frac{\underset{\x\in\XX_{-}}{\sum} \mathbbm{1}
        \bigg[
        \underset{\z'\in\ZZ}{\min}~c_R(\x, \z'; \alpha) < \min\big(1,~\underset{\z''\in\ZZ}{\min}~c_M(\x, \z'')\big) 
        \bigg]}{|\XX_{-}|}.    
\end{align*}

The key observation here is that with subsidy $\alpha$, the recourse cost reduces, but the manipulation cost remains the same. Both optimal recourse actions $\x_R(\x)$ and the optimal manipulation action $\x_M(\x)$ remain the same. With that, we can show that the recourse rate is a monotonic function in subsidy -- as the subsidy level increases, the recourse rate will also increase:
\begin{theorem*}(Subsidy Influence on Recourse Rate)
    \label{thm:sub-rec-rate}
    Given a reveal set $\ZZ$, the recourse rate $\rec(\ZZ, \XX_{-}, \alpha)$ is a monotonically 
    increasing
    function of subsidies $\alpha$.
\end{theorem*}

With subsidy $\alpha$, the social cost for a given revealed set $\ZZ$ becomes:
\begin{align*}
      \resizebox{.95\linewidth}{!}{$\text{cost}(\ZZ, \XX_{-}; \alpha)
      = \sum_{\x\in \XX_{-}}\big(c_R(\x, \z_R(\x, \ZZ; \alpha);\alpha) - c_R(\x, \x_R; \alpha)\big)$}
    \end{align*}
where
$\z_R(\x, \ZZ; \alpha) = \argmin_{\z\in \ZZ} (1 - \alpha)c_R(\x, \z)$ is the optimal recourse action given revealed set $\ZZ$ and a particular subsidy level $\alpha$, and  $\x_R$ is the optimal default recourse action provided by the system without any strategic withholding. We can show that the social cost is also a monotonic non-increasing function in the subsidy level:

\begin{theorem}(Subsidy Influence on Social Cost)
    \label{thm:sub-social-cost}
    Given a revealed set $\ZZ$, 
    the social cost $\text{cost}(\ZZ, \XX_{-}; \alpha)$ is monotonically 
    decreasing
    in subsidies.
\end{theorem}

Subsidies also help improve the system's utility; under some assumptions on the cost functions (i.e., monotonic in the distance and only cross once), the system's utility is monotonic in subsidies as well:

\begin{theorem}(Subsidy's Influence on System's Utility)
\label{thm:sub-utility}
    Given a revealed set $\ZZ$, when both $c_R(\x, \x')$ and $c_M(\x, \x')$ are monotonic in $\|\x - \x'\|$ and only cross once, the system utility is monotonically increasing in subsidies.
\end{theorem}

Next we examine the difference in social cost between groups as a function of subsidies.
We find that subsidies are an effective tool to mitigate disparities caused by strategically withheld recourse.
\begin{theorem}(Subsidy Influence on Social Cost Disparity)
    \label{thm:sub-social-cost-diff}
    With subsidy $\alpha$, the disparity in social cost for two group $g_0, g_1$ becomes: $\resizebox{.98\linewidth}{!}{$\textit{Diff}^{(\text{cost})}(\ZZ, \XX^{(g_0)}, \XX^{(g_1)};\alpha)
         := \left|\text{cost}(\ZZ, \XX_{-}^{(g_1)};\alpha) - \text{cost}(\ZZ, \XX_{-}^{(g_0)};\alpha)\right|$}$
Given a revealed set $\ZZ$, 
the social cost difference monotonically decreases in subsidies. 
\end{theorem}

Intuitively, as we increase the subsidy level, the cost of recourse decreases linearly as a function of the subsidy level, making it increasingly cheaper to perform the optimal recourse action. 
For both social groups, their social cost approaches 0 as we increase the subsidy level; as a result, the disparity in social cost between the two groups also decreases to 0.

With subsidy $\alpha$, for a given a revealed set $\ZZ$, the disparity in recourse ratio for groups $g_0, g_1$ is: 
\[
\resizebox{.98\linewidth}{!}{$ \textit{Diff}^{(\rec)}(\ZZ, \XX^{(g_0)}, \XX^{(g_1)}; \alpha):= \left|\rec(\ZZ, \XX_{-}^{(g_1)}; \alpha) - \rec(\ZZ, \XX_{-}^{(g_0)}; \alpha)\right|$}
\]
where $\rec(\ZZ, \XX_{-}^{(g_i)})$ is the recourse rate for a particular subgroup $g_i$.  We show that when subsidies are sufficiently large, the recourse rate difference is monotonically decreasing in subsidies: 

\begin{theorem}(Subsidy's Influence on Recourse Rate Disparity)
\label{thm:sub-recourse-rate-diff}
   Given two groups $g_0$ and $g_1$ of relatively equal negatively classified agents size $|\XX_{-}^{(g_0)}|\approx |\XX_{-}^{(g_1)}|$, there exists a subsidy level $0 \leq \alpha^*\leq 1$, such that $\forall \alpha \geq \alpha^*$, the recourse rate difference monotonically decreases.
\end{theorem}

This result follows that when recourse is free, i.e., subsidies are maximized, all agents can perform recourse, and the recourse rate difference is $0$. 
Thus, as subsidies increase, there must exist a point (namely $\alpha^*$) when both groups can take advantage of subsidies at proportional rates, thus decreasing the gap between the number of agents performing recourse in both groups. We also verify empirically that for recourse rate difference, there indeed exists a peak subsidy value $\alpha^*$ where the recourse rate difference increases before and then decreases afterward (see Figure \ref{fig:lr_rec_rate_diff_sub}).

\section{EMPIRICAL STUDIES}
\label{sec:experiments}
\paragraph{Setup}
We conduct experiments using three datasets: 1) \textbf{Law School}  \cite{wightman1998lsac} dataset, in which the objective is to predict whether a student will pass the bar exam on the first attempt, \textbf{Adult Income} \cite{dua2017uci} in which the objective is to predict whether an individual earns more than $50K$ annually, and \textbf{German Credit} \cite{yeh2009comparisons} in which the objective is to predict whether a given individual will \emph{not} default on their credit. In each dataset, agents have constant utility over approved features, i.e., the conventional recourse setting where $u_a(\x) = 1$ for all $\x$; the principal (system) has utility $u_p(\x) = 1$ when the agent is a true positive ($y = 1$, $f(x) = 1$) and $u_p(\x) = -1$ when the agent is a false positive ($y = -1$, $f(x) = 1$). 
Qualification is predicted via Logistic Regression (shown in this section) or Gradient Boosting Trees (shown in the Supplement \cref{sec:additional-experimental-results}). 

Recourse and manipulation both carry an $\ell_2$ cost, namely
$c_R (\x, \z) = \|w_R \cdot (\x - \z)\|_2$, and $c_M (\x, \z) = \|w_M \cdot (\x - \z)\|_2$,where $w_R$ and $w_M$ are the weight vectors for the cost functions.
In our experiments, we report outcomes over $100$ runs using randomly initialized $w_R$ and $w_M$ and resampled subsets of positive and negative agents in the dataset in each run. 
We set the probability that the agent discloses their feature publicly at $p=0.7$ for all experiments. When varying this value, we observe similar results.

\begin{figure*}[tbh!]
\centering
\includegraphics[width=0.8\textwidth]{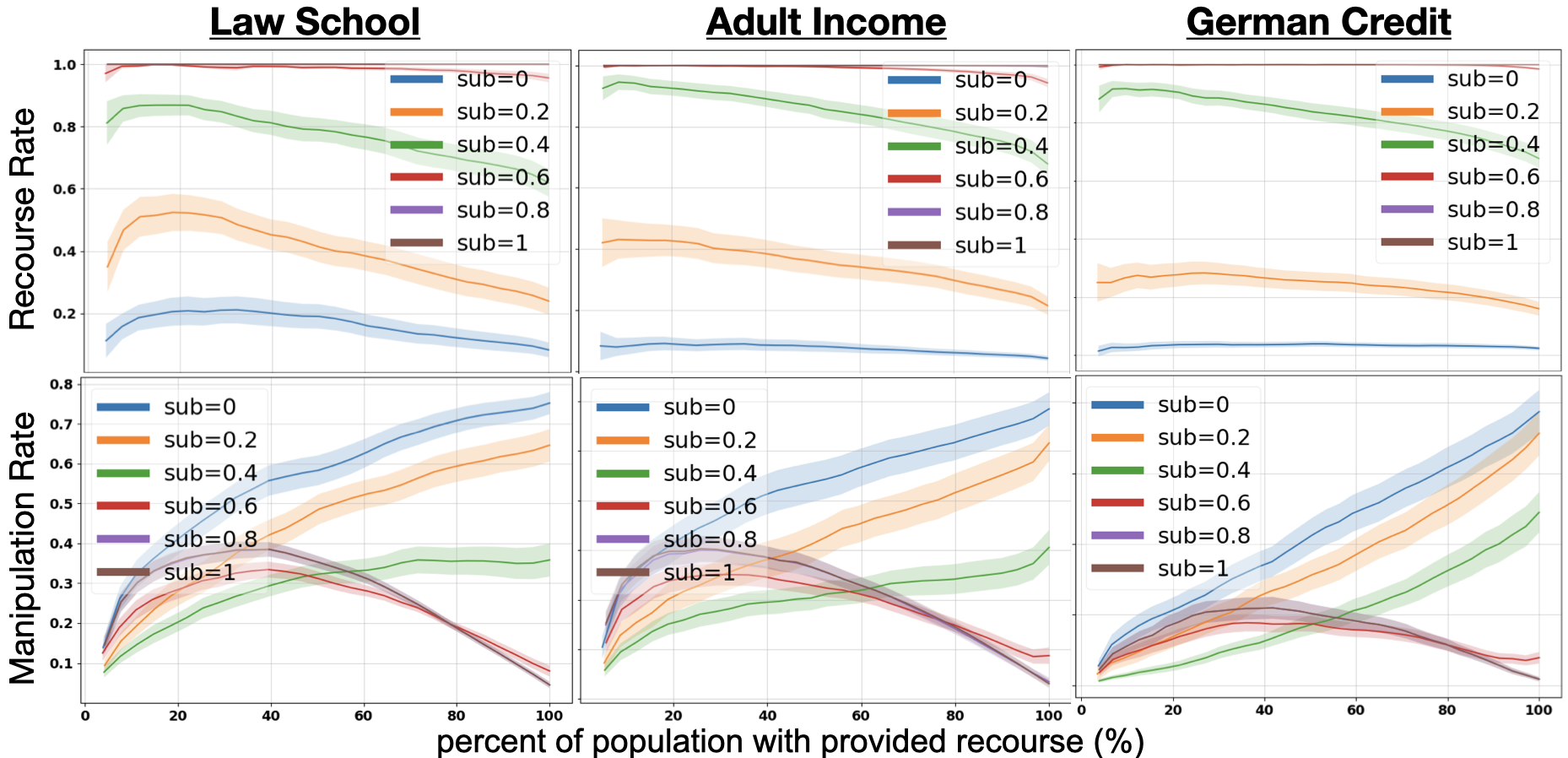}
        \caption{Fraction of the population performing recourse (top row) or manipulation (bottom row). Each line corresponds to a different subsidy ratio ``sub", i.e., the cost reduction applied to recourse. }
        \label{fig:lr_rec_man_ratio}   
\end{figure*}

\begin{figure*}[tbh!]
\centering
        \includegraphics[width = 0.8\textwidth]{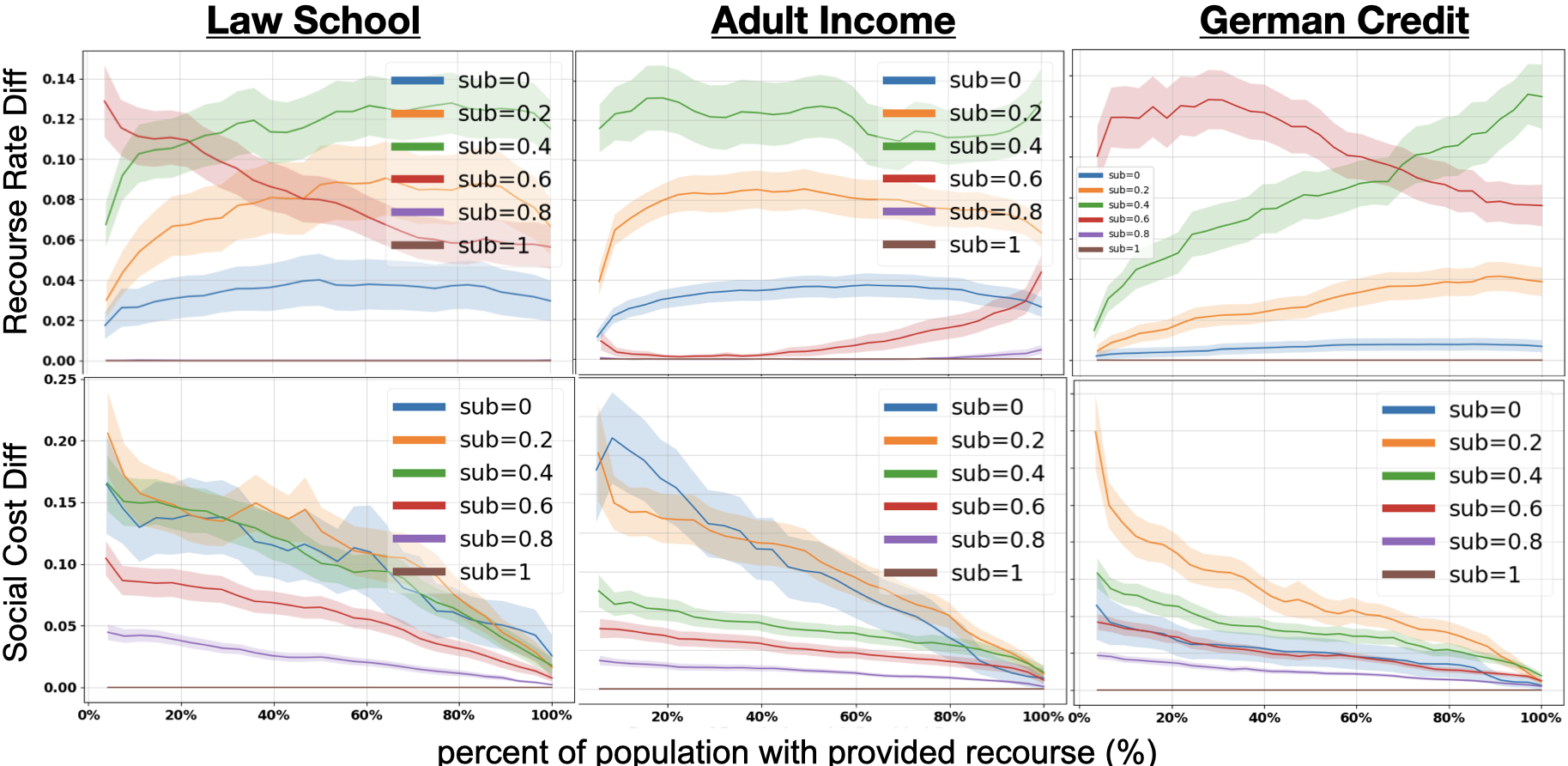}
        \caption{Difference in recourse rate (top row) and social cost (bottom row) between different sensitive attribute groups. Each line corresponds to a different subsidy ratio ``subs", i.e., the cost reduction applied to recourse. }
        \label{fig:lr-rec-social-cost-diff}
\end{figure*}

\paragraph{Recourse Rate and Manipulation Rate} We begin by examining the relationship between the fraction of the population choosing to perform recourse and the fraction choosing to perform manipulation as a function of the fraction of agents given a recourse action.
In Figure \ref{fig:lr_rec_man_ratio}, 
we see that in general, as the percentage of agents who are provided a recourse action increases, the recourse rate decreases while the manipulation rate increases (this trend holds for each subsidy value). 
Thus, when agents themselves can strategically select between recourse and manipulation, the increased model transparency, created by providing more agents with recourse actions, results in more agents selecting to perform manipulation.
Providing more recourse actions to agents, does not necessarily result in more agents performing recourse. Despite this general trend, we also observe the effectiveness of subsidies. 
As subsidies converge to $1$ (meaning recourse carries no cost), the fraction of agents choosing recourse converges to $1$, while the fraction of agents choosing manipulation converges to $0$. 
While it may be expensive in general to provide such subsidies, and the question of how to balance this expense against the system's own utility remains open, these results indicate that subsidies are an effective avenue for broadly promoting recourse and disincentivizing manipulation.

\paragraph{Disparity in Recourse and Social Cost} 
Lastly, we investigate how strategic system behavior causes disparate impacts among sensitive groups. 
In our experiments, groups are taken to be binary and are defined by race in the Law School dataset (White and Non-White), by gender in the Adult Income dataset (Male and Female), and by age in the German Credit dataset (Young and Old).
In Figure \ref{fig:lr-rec-social-cost-diff}, we see the difference in the number of agents performing recourse and social cost between groups.
Higher values in these plots indicate higher rates of recourse, or lower cost, for White individuals in the Law School dataset, Male individuals in the Adult income dataset, and Young individuals in the Credit dataset. 
First, strong subsidies (particularly subs $\leq 0.4$) result in a large decrease in the disparities between groups for both recourse rate and social cost. 
For less strong subsidies (subs $\geq 0.6$), we see that the gap in recourse rate between groups can increase. 
This is due to the fact that when subsidies are less strong, only agents with already low costs of recourse (primarily from the advantaged group) can benefit from those subsidies. 

\section{CONCLUSION}
In scenarios where agents can manipulate a system, there is a reduced incentive for the system to provide recourse due to increased model transparency. Consequently, the system strategically withholds recourse from some, leading to higher social costs, and disproportionately impacting disadvantaged groups. Despite the inherent tension between the system's utility and its provision of recourse, subsidies emerge as a viable tool to boost recourse-providing rates and alleviate group-wise disparities resulting from recourse withholding.

\paragraph{Acknowledgements}
This work is partially supported by the National Science Foundation (NSF) under grants IIS-2214141, IIS-1905558, CNS-2310470, IIS-2143895, IIS-2040800, and CCF-2023495; the Office of Naval Research (ONR) under grant N00014-24-1-2663, and Amazon.

\bibliography{iclr2025_conference}

\begin{thebibliography}{32}
\providecommand{\natexlab}[1]{#1}
\providecommand{\url}[1]{\texttt{#1}}
\expandafter\ifx\csname urlstyle\endcsname\relax
  \providecommand{\doi}[1]{doi: #1}\else
  \providecommand{\doi}{doi: \begingroup \urlstyle{rm}\Url}\fi

\bibitem[Akyol et~al.(2016)Akyol, Langbort, and Basar]{akyol2016price}
Emrah Akyol, Cedric Langbort, and Tamer Basar.
\newblock Price of transparency in strategic machine learning.
\newblock \emph{arXiv preprint arXiv:1610.08210}, 2016.

\bibitem[Barsotti et~al.(2022)Barsotti, Ko{\c{c}}er, and Santos]{barsotti2022transparency}
Flavia Barsotti, R{\"u}ya~G{\"o}khan Ko{\c{c}}er, and Fernando~P Santos.
\newblock Transparency, detection and imitation in strategic classification.
\newblock In \emph{Proceedings of the 31st International Joint Conference on Artificial Intelligence, IJCAI 2022}. International Joint Conferences on Artificial Intelligence (IJCAI), 2022.

\bibitem[Bechavod et~al.(2022)Bechavod, Podimata, Wu, and Ziani]{bechavod2022information}
Yahav Bechavod, Chara Podimata, Steven Wu, and Juba Ziani.
\newblock Information discrepancy in strategic learning.
\newblock In \emph{International Conference on Machine Learning}, pp.\  1691--1715. PMLR, 2022.

\bibitem[Blocki et~al.(2013)Blocki, Christin, Datta, Procaccia, and Sinha]{blocki2013audit}
Jeremiah Blocki, Nicolas Christin, Anupam Datta, Ariel~D Procaccia, and Arunesh Sinha.
\newblock Audit games.
\newblock \emph{arXiv preprint arXiv:1303.0356}, 2013.

\bibitem[Chen et~al.(2020)Chen, Wang, and Liu]{chen2020linear}
Yatong Chen, Jialu Wang, and Yang Liu.
\newblock Linear classifiers that encourage constructive adaptation.
\newblock \emph{arXiv preprint arXiv:2011.00355}, 2020.

\bibitem[Chen et~al.(2023)Chen, Tang, Zhang, and Liu]{chen2023model}
Yatong Chen, Zeyu Tang, Kun Zhang, and Yang Liu.
\newblock Model transferability with responsive decision subjects.
\newblock In \emph{International Conference on Machine Learning}, pp.\  4921--4952. PMLR, 2023.

\bibitem[Chen et~al.(2018)Chen, Podimata, Procaccia, and Shah]{chen2018strategyproof}
Yiling Chen, Chara Podimata, Ariel~D Procaccia, and Nisarg Shah.
\newblock Strategyproof linear regression in high dimensions.
\newblock In \emph{Proceedings of the 2018 ACM Conference on Economics and Computation}, pp.\  9--26, 2018.

\bibitem[Dong et~al.(2018)Dong, Roth, Schutzman, Waggoner, and Wu]{dong2018strategic}
Jinshuo Dong, Aaron Roth, Zachary Schutzman, Bo~Waggoner, and Zhiwei~Steven Wu.
\newblock Strategic classification from revealed preferences.
\newblock In \emph{Proceedings of the 2018 ACM Conference on Economics and Computation}, pp.\  55--70, 2018.

\bibitem[Dua et~al.(2017)Dua, Graff, et~al.]{dua2017uci}
Dheeru Dua, Casey Graff, et~al.
\newblock Uci machine learning repository.
\newblock 2017.

\bibitem[Ehyaei et~al.(2023)Ehyaei, Karimi, Sch{\"o}lkopf, and Maghsudi]{ehyaei2023robustness}
Ahmad-Reza Ehyaei, Amir-Hossein Karimi, Bernhard Sch{\"o}lkopf, and Setareh Maghsudi.
\newblock Robustness implies fairness in causal algorithmic recourse.
\newblock In \emph{Proceedings of the 2023 ACM Conference on Fairness, Accountability, and Transparency}, pp.\  984--1001, 2023.

\bibitem[Estornell et~al.(2023{\natexlab{a}})Estornell, Chen, Das, Liu, and Vorobeychik]{estornell2023incentivizing}
Andrew Estornell, Yatong Chen, Sanmay Das, Yang Liu, and Yevgeniy Vorobeychik.
\newblock Incentivizing recourse through auditing in strategic classification.
\newblock In \emph{Proceedings of the Thirty-Second International Joint Conference on Artificial Intelligence}, pp.\  400--408, 08 2023{\natexlab{a}}.
\newblock \doi{10.24963/ijcai.2023/45}.

\bibitem[Estornell et~al.(2023{\natexlab{b}})Estornell, Das, Liu, and Vorobeychik]{estornell2021unfairness}
Andrew Estornell, Sanmay Das, Yang Liu, and Yevgeniy Vorobeychik.
\newblock Group-fair classification with strategic agents.
\newblock In \emph{ACM Conference on Fairness, Accountability, and Transparency}, pp.\  389--399, 2023{\natexlab{b}}.

\bibitem[Fokkema et~al.(2024)Fokkema, Garreau, and van Erven]{fokkema2024risks}
Hidde Fokkema, Damien Garreau, and Tim van Erven.
\newblock The risks of recourse in binary classification.
\newblock In \emph{International Conference on Artificial Intelligence and Statistics}, pp.\  550--558. PMLR, 2024.

\bibitem[Gupta et~al.(2019)Gupta, Nokhiz, Roy, and Venkatasubramanian]{gupta2019equalizing}
Vivek Gupta, Pegah Nokhiz, Chitradeep~Dutta Roy, and Suresh Venkatasubramanian.
\newblock Equalizing recourse across groups.
\newblock \emph{arXiv preprint arXiv:1909.03166}, 2019.

\bibitem[Hardt et~al.(2016)Hardt, Megiddo, Papadimitriou, and Wootters]{hardt2016strategic}
Moritz Hardt, Nimrod Megiddo, Christos Papadimitriou, and Mary Wootters.
\newblock Strategic classification.
\newblock In \emph{Proceedings of the 2016 ACM conference on innovations in theoretical computer science}, pp.\  111--122, 2016.

\bibitem[Harris et~al.(2022)Harris, Chen, Kim, Talwalkar, Heidari, and Wu]{harris2022bayesian}
Keegan Harris, Valerie Chen, Joon Kim, Ameet Talwalkar, Hoda Heidari, and Steven~Z Wu.
\newblock Bayesian persuasion for algorithmic recourse.
\newblock \emph{Advances in Neural Information Processing Systems}, 35:\penalty0 11131--11144, 2022.

\bibitem[Hu et~al.(2019)Hu, Immorlica, and Vaughan]{hu2019disparate}
Lily Hu, Nicole Immorlica, and Jennifer~Wortman Vaughan.
\newblock The disparate effects of strategic manipulation.
\newblock In \emph{Proceedings of the Conference on Fairness, Accountability, and Transparency}, pp.\  259--268, 2019.

\bibitem[Karimi et~al.(2020{\natexlab{a}})Karimi, Barthe, Schölkopf, and Valera]{karimi2020survey}
Amir-Hossein Karimi, Gilles Barthe, Bernhard Schölkopf, and Isabel Valera.
\newblock A survey of algorithmic recourse: definitions, formulations, solutions, and prospects, 2020{\natexlab{a}}.

\bibitem[Karimi et~al.(2020{\natexlab{b}})Karimi, von Kügelgen, Schölkopf, and Valera]{karimi2020algorithmic}
Amir-Hossein Karimi, Julius von Kügelgen, Bernhard Schölkopf, and Isabel Valera.
\newblock Algorithmic recourse under imperfect causal knowledge: a probabilistic approach, 2020{\natexlab{b}}.

\bibitem[Levanon \& Rosenfeld(2021)Levanon and Rosenfeld]{levanon2021strategic}
Sagi Levanon and Nir Rosenfeld.
\newblock Strategic classification made practical.
\newblock In \emph{International Conference on Machine Learning}, pp.\  6243--6253. PMLR, 2021.

\bibitem[Milli et~al.(2019)Milli, Miller, Dragan, and Hardt]{milli2019social}
Smitha Milli, John Miller, Anca~D Dragan, and Moritz Hardt.
\newblock The social cost of strategic classification.
\newblock In \emph{Proceedings of the Conference on Fairness, Accountability, and Transparency}, pp.\  230--239, 2019.

\bibitem[Olckers \& Walsh(2023)Olckers and Walsh]{olckers2023incentives}
Matthew Olckers and Toby Walsh.
\newblock Incentives to offer algorithmic recourse, 2023.

\bibitem[Orso et~al.(2015)Orso, Lee, and Shen]{orso2015submodular}
Andrew Orso, Jon Lee, and Siqian Shen.
\newblock Submodular minimization in the context of modern lp and milp methods and solvers.
\newblock In \emph{Proceedings of the 14th International Symposium on Experimental Algorithms - Volume 9125}, pp.\  193–204, Berlin, Heidelberg, 2015. Springer-Verlag.
\newblock ISBN 9783319200859.

\bibitem[{Tsirtsis} et~al.(2019){Tsirtsis}, {Tabibian}, {Khajehnejad}, {Singla}, {Sch{\"o}lkopf}, and {Gomez-Rodriguez}]{tsirtsis2019optimal}
Stratis {Tsirtsis}, Behzad {Tabibian}, Moein {Khajehnejad}, Adish {Singla}, Bernhard {Sch{\"o}lkopf}, and Manuel {Gomez-Rodriguez}.
\newblock {Optimal Decision Making Under Strategic Behavior}.
\newblock \emph{arXiv e-prints}, 2019.

\bibitem[Ustun et~al.(2019)Ustun, Spangher, and Liu]{ustun2019actionable}
Berk Ustun, Alexander Spangher, and Yang Liu.
\newblock Actionable recourse in linear classification.
\newblock In \emph{Proceedings of the conference on fairness, accountability, and transparency}, pp.\  10--19, 2019.

\bibitem[Venkatasubramanian \& Alfano(2020)Venkatasubramanian and Alfano]{venkatasubramanian2020philosophical}
Suresh Venkatasubramanian and Mark Alfano.
\newblock The philosophical basis of algorithmic recourse.
\newblock In \emph{Proceedings of the 2020 Conference on Fairness, Accountability, and Transparency}, pp.\  284--293, 2020.

\bibitem[von Kügelgen et~al.(2020)von Kügelgen, Bhatt, Karimi, Valera, Weller, and Schölkopf]{vonkugelgen2020fairness}
Julius von Kügelgen, Umang Bhatt, Amir-Hossein Karimi, Isabel Valera, Adrian Weller, and Bernhard Schölkopf.
\newblock On the fairness of causal algorithmic recourse, 2020.

\bibitem[von Kügelgen et~al.(2022)von Kügelgen, Karimi, Bhatt, Valera, Weller, and Schölkopf]{vonkügelgen2022fairness}
Julius von Kügelgen, Amir-Hossein Karimi, Umang Bhatt, Isabel Valera, Adrian Weller, and Bernhard Schölkopf.
\newblock On the fairness of causal algorithmic recourse, 2022.

\bibitem[Vorobeychik(2023)]{vorobeychik2023many}
Yevgeniy Vorobeychik.
\newblock The many faces of adversarial machine learning.
\newblock In \emph{AAAI Conference on Artificial Intelligence}, volume~37, pp.\  15402--15409, 2023.

\bibitem[Wightman \& Council(1998)Wightman and Council]{wightman1998lsac}
L.F. Wightman and Law School~Admission Council.
\newblock \emph{LSAC National Longitudinal Bar Passage Study}.
\newblock LSAC research report series. Law School Admission Council, 1998.
\newblock URL \url{https://books.google.com/books?id=O9A7AQAAIAAJ}.

\bibitem[Yeh \& Lien(2009)Yeh and Lien]{yeh2009comparisons}
I-Cheng Yeh and Che-hui Lien.
\newblock The comparisons of data mining techniques for the predictive accuracy of probability of default of credit card clients.
\newblock \emph{Expert systems with applications}, 36\penalty0 (2):\penalty0 2473--2480, 2009.

\bibitem[Zrnic et~al.(2021)Zrnic, Mazumdar, Sastry, and Jordan]{zrnic2021leads}
Tijana Zrnic, Eric Mazumdar, Shankar Sastry, and Michael Jordan.
\newblock Who leads and who follows in strategic classification?
\newblock \emph{Advances in Neural Information Processing Systems}, 34:\penalty0 15257--15269, 2021.

\end{thebibliography}
\bibliographystyle{iclr2025_conference}





\newpage

\appendix
\newpage
\onecolumn
\aistatstitle{To Give or Not to Give? The Impacts of Strategically Withheld Recourse}

\section{Notation Table}
\label{secA:notation}
\begin{table}[!ht]
    \centering
    \begin{tabular}{c l}
        \toprule
            Symbol & Usage \\
        \midrule
            $\X\subset \mathbb{R}^d$  &  The domain of the feature $\x$  \\
            $\Y\equiv\{0, 1\}$  & The domain of labels \\ 
            $\XX \in \mathbb{R}^{n\times d} $ & A set of features of $n$ agents \\
            $\YY\in \{0, 1\}^{|\XX|}$ & The labels for the set of features $\XX$ \\
            $\x\in \X$ & A random variable representing an example's features\\
            $y\in \Y$ & A random variable representing an example's \emph{ground truth label} \\
            $f:\X\rightarrow \Y$ & a binary classifier, unknown to the agents \\
            $\X_{-}, \X^{(0)}\subseteq \X$ &  The domain of negatively classified features,  i.e. $\forall \x\in\X_{-}$, $f(\x) = 0$ \\
            $\X_{+}\subseteq \X$ &  The domain of positively classified features, i.e., $\forall \x \in \X_{+}, f(\x) = 1$  \\
            $\XX_{-}\subseteq \X$ &  The set of negatively classified features,  i.e. $\forall \x\in\XX_{-}$, $f(\x) = 0$ \\
            $\XX_{+}\subseteq \X$ &  The set of positively classified features, i.e., $\forall \x \in \XX_{+}, f(\x) = 1$  \\
            $\XX^{(g_i)} \subseteq \XX$ & The subset of features belongs to group $G = g_i$  \\
            $c_R: \X \times \X \rightarrow \R_{+}$ & The cost function of recourse \\
            $c_M: \X \times \X \rightarrow \R_{+}$ & The cost function of recourse \\
            $\XX_R$ & The set of all possible recourse actions \\  
            $\ZZ_R$ & The set of revealed recourse actions \\ 
            $\ZZ_+$ & The set of revealed positively classified features \\ 
            $\ZZ = \ZZ_R \cup \ZZ_+$ & A publicly revealed feature set \\
            $\x_R(\x)$ & The optimal recourse action for agent with feature $\x$\\
            $\x_M(\x)$ & The optimal manipulation action for agent with feature $\x$\\
            $\z(\x, \ZZ)$ & The agent's final action\\
            $\rec(\ZZ,\XX)$ & The recourse ratio for feature sets $\XX$ given revealed set is $\ZZ$\\
            $\alpha\in [0, 1]$ & A subsidy level\\
            $u_0 \in \mathbb{R}$ & The initial utility of a system without providing recourse\\  
        \bottomrule
    \end{tabular}
    \caption{Primary Notation}
    \label{tab:notation_table}
  \end{table}

\newpage

\section{Additional Related Work}
\label{sec:additional-related-work}
\paragraph{Recourse} Recourse focuses on providing agents with the ability to contest or improve their outcome via a modification to their attributes in a \emph{genuine} manner (e.g., paying off debt to increase creditworthiness) \cite{ustun2019actionable, venkatasubramanian2020philosophical,karimi2020survey,gupta2019equalizing,karimi2020algorithmic, vonkugelgen2020fairness,chen2020linear, harris2022bayesian}. Much of this line of work focuses on the setting where the requested recourse is guaranteed to be provided. 
As far as we know, our work is the first to challenge this fundamental assumption and argue that without a third-party's intervention (e.g., the government regulation on the system's recourse providing), a utility-maximizing algorithmic recourse system may be incentivized to strategically withhold recourse from some agents to prevent manipulations. We point the reader to \cite{karimi2020survey} for a more detailed discussion of the concepts and recent development of algorithmic recourse. To our knowledge, even though the literature has previously introduced the concepts of recourse, strategic manipulation, and subsidy analysis, we have yet to find any studies that explicitly highlight how a recourse system might act strategically by withholding recourse to enhance its own utility. The originality of our work, thus, is to address this gap. 
Some works have investigated the relationship between incentives/utility and recourse, such as \cite{fokkema2024risks}, which finds that providing recourse can decrease classifier accuracy, \cite{estornell2023incentivizing}, which investigates ways to ensure that given recourse actions are taken by agents, and \cite{olckers2023incentives} which investigates the incentive compatibility of recourse.

\paragraph{Strategic Classification} Strategic Classification focuses on the problem of how to effectively make predictions in the presence of agents who behave strategically to obtain desirable outcomes \cite{hardt2016strategic,chen2018strategyproof,tsirtsis2019optimal,levanon2021strategic,dong2018strategic,chen2018strategyproof, zrnic2021leads,chen2023model}. 
Our work considers a specific type of strategic behavior, namely the \emph{imitation-based} manipulations: 
agents do not know the classifier $f$ but are aware of a set of positively classified features and can misreport their feature by imitating another agent's feature that is positively classified. Such copycat behavior has been well-known in the literature of game theory, the behavioral economy, and strategic classification, e.g., \cite{bechavod2022information, barsotti2022transparency}. 
While most of this line of work focuses on agents being strategic and could potentially modify their features to get a favorable prediction outcome, our work focuses on when the system is being strategic and potentially withholds recourse to the agents. 

\paragraph{Fairness and Social Cost in Recourse and Strategic Classification} Fairness has been explored in the literature algorithmic recourse and strategic classification. 
For example, existing works on fairness in recourse emphasize the importance of equitable recourse and explore various remedying unfair recourse decisions \cite{gupta2019equalizing, vonkügelgen2022fairness, ehyaei2023robustness}. 
Fairness with the presence of strategic behavior has featured studies that highlight the inequity that results from strategic behavior by individuals \cite{hu2019disparate}, as well as inequity (e.g., social cost) resulting from making classifiers robust to strategic behavior \cite{milli2019social, estornell2021unfairness}. 
Unlike previous work that primarily focuses on proposing fair classifiers with the presence of strategic agents, our work uniquely demonstrates how the system's strategic withholding impacts the fairness and social cost for different societal groups. 

\paragraph{Transparency} Also related is work on transparency in machine learning. 
In particular, 
\cite{barsotti2022transparency} find that the risks of transparent explanations are alleviated if effective methods to detect faking behaviors are in place. Unlike our modeling framework, they model transparency as how much noise is in the threshold of a threshold classifier. 
\cite{akyol2016price} examines the impact of users' strategic behavior on the design and performance of transparent machine learning algorithms, quantifying the "price of transparency" as the cost ratio for the algorithm designer when users exploit transparency compared to when the algorithm is opaque.

\paragraph{Comparison with three closely related papers}
\begin{itemize}
    \item Comparison with \cite{estornell2023incentivizing}: the key distinction between these our work and \cite{estornell2023incentivizing} is that \cite{estornell2023incentivizing} presumes the system will provide any agent with an optimal recourse action and examines how auditing can dissuade agents from manipulating. In contrast, we do not consider auditing, and instead focus on how a system may be incentivized to withhold recourse from certain agents. While both papers examine the use of subsidies, our focus and model are distinct.
    \item Comparison with \cite{fokkema2024risks}: \cite{fokkema2024risks} focuses on discussing the accuracy drop as a result of the system providing recourse, because it pushes users to regions of higher class uncertainty and therefore leads to more mistakes. Our work, on the other hand, focuses on the incentive-compatibility problem in an algorithmic recourse system. 
of the population
    \item Comparison with \cite{olckers2023incentives}: similar to our work, \cite{olckers2023incentives} also studies when it is incentive-compatible for a decision-maker to offer recourse. Unlike our setting, however, they primarily operate on a simple toy model that assumes the applicant’s \emph{profitability} is fixed.
\end{itemize}

\section{Proofs for Section \ref{sec:system}}

\subsection{ILP for system when p = 1}
\label{subsec:ILP}
We provide the ILP formula for the system to find optimal recourse actions when the revealing probability $p =1$:
\begin{align}
    \max_{\mathbf{a}\in\{0, 1\}^{|\ZZ_{\max}|}, \mathbf{b}\in\{0, 1\}^{|\XX_{-}|}}&~\sum_{j=1}^{|\XX_{-}|} b_j && \tag{maximize the number of agents performing recourse}\\
    \text{s.t. }&~b_j c_R(\x_j, \z_R) \leq a_i c_M(\x_j, \z_i) + (1-a_i) && \tag{only do recourse if all manipulation costs are greater}\\
    &~b_j \leq a_{j_R} && \tag{the optimal recourse action $\z_{j_R}$ for agent $j$ must be revealed}\\
    &~b_j c_R(x_j, z_R) \leq 1 && \tag{the optimal recourse action $\z_{j_R}$ for agent $j$ must be less than 1}\\
    & \sum_{h = 1}^{|\ZZ|} a_{h} = k && \tag{the total number of revealed recourse action is $k$}
\end{align}

\subsection{NP hardness of the System's Optimal Recourse Providing Problem}
\label{sec:np-hard}

\begin{theorem}
\label{thm:nonseq_hard}
    The problem of selecting the optimal set of recourse actions to recommend, such that the system's utility is maximized (Equation \ref{eq:sys_obj}), is NP-hard, even when the probability of diclosure $p = 1$.
\end{theorem}

\begin{proof}
\label{proof:nonseq_hard}
    To demonstrate the intractability of this objective, we reduce from the known NP-hard problem Minimum $k$-Union (M$k$U), an instance of which is defined via a universe of $n$ elements $U=\{s_1, \ldots, s_n\}$,a collection of $n$ sets $\textbf{S} = \{S_1, \ldots S_m\}$ with elements in $U$, and a budget $k$. 
    The objective in M$k$U is to select an index set $I$ of size exactly $k$ such that $\cup_{j\in I} S_j$ is minimized.
    Given an instance of M$k$U can be mapped to an instance of simultaneous recourse as follows.
    Let $\XX^{(0)}\times \ZZ = \big\{(\x, \z_j): s_i\in U \text{ and } S_j\in \mathbf{S} \big\}$, and define $c_R$ and $c_M$ as follows,
    \begin{align*}
            c_R(\x, \z_j) = \begin{cases}
                                    1 & \text{ if }~ i\neq j \\
                                    0 & \text{ if }~ i=j
                              \end{cases}\qquad
            c_M(\x, \z_j) = \begin{cases}
                                    1 & \text{ if }~ s_i\notin S_j \\
                                    1/2 & \text{ if }~ s_i\in S_j
                              \end{cases}
    \end{align*}
    Under this construction of the cost functions, each agent $\x$ will perform recourse if and only if $\z_i$ is revealed, and the disclosure probability $p = 1$.
    In the case that $\z_i$ is not revealed, the agent will elect to perform manipulation when any $\z_j$ is revealed where $j\neq i$ and $s_i\in S_j$.
    If neither criterion is met, the agent will elect to do nothing (remaining negatively classified).
    Combining these cases, we see that revealing each $\z_j$ causes exactly one agent to perform recourse, namely $\x_j$, and causes all $\x$ (with $s_i\in S_j$) to manipulate. 
    Let $I = \{j_1, \ldots, j_k\}$ be the index set of the revealed features, then the number of agents manipulating is equal to $\big|\cup_{j\in I} S_j \big| - k$. 
    Therefore providing $k$ recourse actions to agents while minimizing the number of agents manipulating is equivalent to minimizing  $\big|\cup_{j\in I} S_j \big|$.
\end{proof}

\subsection{Submodularity of the System's Utility}
\label{sec:submodular}
\begin{theorem}
\label{thm:submodular-objective-function}
The system's objective function is \emph{submodular} with respect to the size of the set of revealed features. 
\end{theorem}

\begin{proof}
Given a revealed set $\ZZ \subseteq \XX_R$, for agent with feature $\x\in \XX_{-}$, let $S_m(\x, \ZZ):= \{z \in \ZZ: c_M(\x, z) \leq c_R(\x, \z_R(\x, \ZZ))\}$ 
be the set of manipulation features that are cheaper than the minimum recourse action $\z_R(x, \ZZ)$ given the revealed set $\ZZ$. 
Then the agent will perform recourse if and only if $S_m(\x, \ZZ) = \emptyset$. 
Given the cost function $c_M$ and $c_R$, the principal can pre-compute each agent's manipulation set $S_m(\x, \ZZ)$. 

The probability for the manipulation set $S_m(\x, \ZZ)$ to overlap with a given revealed set $\ZZ$
is $P(\x; \ZZ) = \Pi_{z\in \ZZ_m(\x, \ZZ)}(1 - p)$, where $p$ is the disclosure probability for any criteria $\z$. 


The goal for the system is to select a disclosure set $\ZZ \subseteq \XX_R$ to \emph{minimize} the overlap between $\ZZ$ and $S_m(\x, \ZZ)$ for all agents, namely:
\begin{align}
\label{eqn:obj-prob-disclosure}
    \min_{\ZZ \subseteq \XX_R} u(\ZZ, \XX_{-})
    & := \sum_{\x\in \XX_{-}} \left(1 - P (\x; \ZZ)\right) = \sum_{\x\in \XX_{-}} \left(1 - \Pi_{z\in \ZZ_m(\x, \ZZ)}(1 - p)\right) 
\end{align}

To ease the notation, we use $u(\ZZ)$ to shorthand $u(\ZZ, \XX_{-})$ since $\XX_{-}$ is fixed in our setting.
To show that \Cref{eqn:obj-prob-disclosure} is submodular, it is equivalent to prove that the objective function $u(\ZZ, \X_{-})$
satisfies the \emph{diminishing returns property}, which means $\forall A, B\subseteq \ZZ$ with $A\subseteq B \subseteq \ZZ$, and any criteria $z\in \ZZ \backslash B$, we want to show
    \begin{align*}
        u(A \cup \{z\}) - u(A) \geq u(B \cup \{z\}) - u(B)
    \end{align*}

Only four types of agents could potentially contribute to the marginal gain for $U$ when the revealed sets are $A \cup \{z\}$  v.s. $B  \cup \{z\}$:
\begin{enumerate}
    \item when $S_m(\x, B \cup \{z\}) = B \cup \{z\}$
    \item when $S_m(\x, B \cup \{z\}) = A \cup \{z\}$
    \item when $S_m(\x, B \cup \{z\}) = \{z\} $
    \item when $S_m(\x, B \cup \{z\}) = B\backslash A \cup \{z\}$
\end{enumerate}

For the first three cases, we can verify that the two marginal gains are the same. For the last case, the two marginal gains are:
    \begin{align*}
        u(A + \{z\}) -u(A) &= \left[1 - (1 - p)\right] - 0 = p\\
        u(B + \{z\}) - u(B) &=
        [1 - \Pi_{t\in \{B\backslash A \cup \{z\}\}}(1 - p)]
        - [1 - \Pi_{t\in \{B\backslash A \}}(1 - p)]\\
        &= p \times \Pi_{t\in \{B\backslash A \}}(1 - p)\\
        &\leq p
    \end{align*}

Since this holds for all agents, we show that adding a criterion $z$ to a larger set $B$ provides an equal or smaller marginal gain in the objective function compared to adding it to a smaller set $A$, satisfying the diminishing returns property. Therefore, the objective function defined in \Cref{eqn:obj-prob-disclosure} is submodular.

\end{proof}

\subsection{Proof for \Cref{thm:expected-utility-change}}

\begin{proof}
\label{proof:thm-expected-utility-change}
    Notice that only agents $\x \in X^{(0)}$ who are originally negatively classified would request a recourse from the system in the first place, and both the recourse action and the manipulation actions that they are potentially going to take will be positively classified by the system.  
    From the system's perspective, when the classifier is non-trivial (better than random guessing), all positively classified $\x$ are more likely to have true label $y = 1$, and all negatively classified $\x$ are more likely to have true label $0$. 
    When an agent with feature $\x$ takes recourse, the expected system utility change is:
    \begin{align*}
        \Delta(\text{System's Utility}) (\x \rightarrow \z_R) &= \big(\mathbbm{1}[y(\z_R) = 1, f(\z_R) = 1] - \mathbbm{1}[y(\z_R)= -1, f(\z_R) = 1]\big) - 0  \\
        &= 2\Pr[y(\z_R) = 1 | X = \z_R] - 1  \geq 0 && \tag{f is a non-trivial classifier, and $f(\z_R) = 1$ }
    \end{align*}
    Similarly, when the agent takes manipulation, the expected system utility change is:
    \begin{align*}
        \Delta(\text{System's Utility}) (\x \rightarrow \z_M) &= 
        \big(\mathbbm{1}[y(\x) = 1, f(\z_M) = 1] - \mathbbm{1}[y(\x)= -1, f(\z_M) = 1]\big) - 0 \\
        &= 2\Pr[y(\x) = 1 | X = \x] - 1 \leq 0  \tag{Since f is a non-trivial classifier, and $f(\x) = 0$}
    \end{align*}
    When the agent performs do-nothing, the system utility remains the same. 
\end{proof}

\section{Proof for Section \ref{sec:cost}}

We first prove a theorem on the monotonicity of social cost.

\begin{theorem*}[Monotonicity of Social Cost]
    When the recourse cost $c_R(x, x')$ is monotonic in $\|x - x'\|$, and consider a linear threshold classifier. The social cost monotonically decreases in the easiest obtained recourse action.
\end{theorem*}

\begin{proof}
\label{proof:thm-social-cost}
Consider a 1-dimensional setting, where the system uses a linear threshold classifier $f(x) = \mathbbm{1}[x\geq \tau]$. In this case, the optimal recourse action for any agent is always the minimum recourse action that has been revealed so far, namely $z_{\min} = \min_{z\in \ZZ} z$.
Recall the definition of the social cost:
    \begin{align*}
      \text{cost}(\ZZ, \XX_{-}) 
      = \sum_{\x\in \XX_{-}}\left(c_R(\x, \z_R(\x, \ZZ) - c_R(\x, \x_R)\right), \text{where}~
        \z_R(\x, \ZZ) = \argmin_{\z\in \ZZ} c_R(\x, \z)
    \end{align*}
When the cost function is monotonic in the $\ell_2$ norm, e.g., $c_R(x, x') = w_R \cdot \|x - x'\|$, we have
\begin{align*}
    c_R(\x, \z_R(\x, \ZZ)) &= w_R\cdot \| x - \z_R(\x, \ZZ) \| = w_R\cdot \min_{z\in \ZZ} \|x - z\| = w_R\cdot \left(\min_{z\in \ZZ} z - x\right)\\
    c_R(\x, \x_R) &= w_R\cdot \|\x - \x_R \| = w_R\cdot \|\x - \tau \| = w_R\cdot \left(\tau - x\right)
\end{align*}
Thus,
\begin{align*}
      \text{cost}(\ZZ, \XX_{-}) 
      &= \sum_{\x\in \XX_{-}}\left(c_R(\x, \z_R(\x, \ZZ)) - c_R(\x, \z_R)\right) \\
      &= \sum_{\x\in \XX_{-}} \left[w_R\cdot \left(\min_{z\in \ZZ} z - x\right) - w_R\cdot \left(\tau - x\right) \right]\\
      &= |\XX_{-}|\cdot w_R\cdot \left(\min_{z\in \XX_{-}} z - \tau \right)
    \end{align*}
As the size of $\ZZ$ gets larger (more recourse actions get revealed), $\min_{z\in \ZZ} z$ will be non-increasing, which means that $\text{cost}(\ZZ, \XX_{-})$ is monotonically decreasing.
    
\end{proof}


\section{Proofs for Section \ref{sec:subsidy}}

\subsection{Proof for Theorem \ref{thm:sub-rec-rate}}

\begin{proof}
    \label{proof:thm-sub-rec-rate}
    Recall that given a revealed set $\ZZ$, with subsidy $\alpha$, the corresponding recourse rate becomes:
    \begin{align*}
        \rec(\ZZ, \XX_{-}; \alpha) = \frac{\underset{\x\in\XX_{-}}{\sum} \mathbbm{1}
        \bigg[
        \underset{\z'\in\ZZ}{\min}~c_R(\x, \z'; \alpha) < \min\big(1,~\underset{\z''\in\ZZ}{\min}~c_M(\x, \z'')\big) 
        \bigg]}{|\XX_{-}|}
    \end{align*}
    In particular, with subsidy $\alpha$, the cost of recourse becomes $(1 - \alpha)\cdot c_R(\x, \z')$, the cost of manipulation remains the same. Both optimal actions $\z_R$ and $\z_M$ remain the same.

    Thus, for the nominator, we have:
    \begin{align*}
        &\underset{\x\in\XX_{-}}{\sum} \mathbbm{1}
        \bigg[
        \underset{\z'\in\ZZ}{\min}~c_R(\x, \z'; \alpha) \leq \min\big(1,~\underset{\z''\in\ZZ}{\min}~c_M(\x, \z'')\big) 
        \bigg] \\
        =& 
        \underset{\x\in\XX_{-}}{\sum} \mathbbm{1}
        \bigg[
        \underset{\z'\in\ZZ}{\min}~ (1 - \alpha)\cdot c_R(\x, \z') \leq \min\big(1,~\underset{\z''\in\ZZ}{\min}~c_M(\x, \z'')\big) 
        \bigg]\\
        =& \underset{\x\in\XX_{-}}{\sum} \mathbbm{1}
        \bigg[(1 - \alpha)\cdot
        \underbrace{\underset{\z'\in\ZZ}{\min}~  c_R(\x, \z')}_{\text{fixed}} \leq \underbrace{\min\big(1,~\underset{\z''\in\ZZ}{\min}~c_M(\x, \z'')}_{\text{fixed}}\big) 
        \bigg]\\
        =& \underset{\x\in\XX_{-}}{\sum} \mathbbm{1}
        \bigg[(1 - \alpha)\cdot
        \underbrace{\underset{\z'\in\ZZ}{\min}~  c_R(\x, \z')}_{\text{fixed for a particular x}} \leq \underbrace{\min\big(1,~\underset{\z''\in\ZZ}{\min}~c_M(\x, \z'')}_{\text{fixed for a particular x}}\big) 
        \bigg]\\
        =& \underset{\x\in\XX_{-}}{\sum} \mathbbm{1}
        \bigg[\alpha\geq \underbrace{1 - \frac{\min\big(1,~\underset{\z''\in\ZZ}{\min}~c_M(\x, \z'')\big)}{\underset{\z'\in\ZZ}{\min}~  c_R(\x, \z')}}_{\text{fixed for a particular x}} 
        \bigg]
    \end{align*}
As $\alpha$ becomes larger, this quantity will be non-decreasing. This implies that the recourse rate is a monotonically non-decreasing function of subsidy for a given revealed set $\ZZ$. 
\end{proof}

\subsection{Proof for \cref{thm:sub-social-cost}}


\begin{proof}
    \label{proof:sub-social-cost}
    Again, consider a 1-dimensional setting, where the system uses a linear threshold classifier $f(x) = \mathbbm{1}[x\geq \tau]$. In this case, the optimal recourse action for any agent is always the minimum recourse actions that has been revealed so far, namely $z_{\min} = \min_{z\in \ZZ} z$.
    Recall the definition of the social cost with subsidy level $\alpha$:
    \begin{align*}
      \text{cost}(\ZZ, \XX_{-}; \alpha) 
      = \sum_{\x\in \XX_{-}}\big(c_R(\x, \z_R(\x, \ZZ; \alpha);\alpha) - c_R(\x, \z_R)\big), \text{where}~
        \z_R(\x, \ZZ; \alpha) = \argmin_{\z\in \ZZ} (1 - \alpha)c_R(\x, \z)
    \end{align*}
    In the 1-dimension case, we have
    \begin{align*}
        c_R(\x, \z_R(\x, \ZZ; \alpha); \alpha) &= (1 - \alpha) \cdot w_R\cdot \| \x - \z_R(\x, \ZZ; \alpha) \| = (1 - \alpha)\cdot w_R\cdot \min_{z\in \ZZ} \|x - z\| = (1 - \alpha)\cdot w_R\cdot \left(\min_{z\in \ZZ} z - x\right)\\
        c_R(\x, \z_R; \alpha) &= (1 - \alpha) \cdot w_R\cdot \|\x - \z_R \| = (1 - \alpha)\cdot w_R\cdot \|\x - \tau \| = (1 - \alpha)\cdot w_R\cdot \left(\tau - x\right)
    \end{align*}
    Thus,
    \begin{align*}
          \text{cost}(\ZZ, \XX_{-};\alpha) 
          &= \sum_{\x\in \XX_{-}}\left(c_R(\x, \z_R(\x; \ZZ; \alpha)) - c_R(\x, \x_R;\alpha)\right) \\
          &= \sum_{\x\in \XX_{-}} \left[(1 - \alpha)\cdot w_R\cdot \left(\min_{z\in \ZZ} z - x\right) - (1 - \alpha)\cdot w_R\cdot \left(\tau - x\right) \right]\\
          &= (1 - \alpha)\cdot |\XX_{-}|\cdot w_R\cdot \left(\min_{z\in \ZZ} z - \tau \right)
    \end{align*}
    As the level of subsidy gets larger ($\alpha$ gets bigger, cheaper to perform recourse), $\text{cost}(\ZZ, \XX_{-}; \alpha)$ will get smaller, which corresponds to a smaller social cost.   
    
\end{proof}

\subsection{Proof for Theorem \ref{thm:sub-utility}}



\begin{proof*}
\label{proof:thm-sub-utility}
The system utility is defined as the difference between true positive and false positive \emph{after} agent's actions. Let $\Pr[Y = 1 | X = x]$ be the true qualification rate given a feature $X = x$, and assume it's also monotonic in $X$. $u_0$ is the system's initial utility (before providing recourse).

Let the recourse region $R_R$ and manipulation region $R_M$ are defined as:
    \begin{align*}
        R_M = \{x\in \XX^{(0)}: c_M(x, z_{\min}) < \min(1, c_R(x, z_{\min}))\}\\
        R_R= \{x\in \XX^{(0)}: c_R(x, z_{\min}) < \min(1, c_M(x, z_{\min}))\}
    \end{align*}
    where $\X^{(0)}$ is the set of negatively classified agents.
    Then we have
\begin{align*}
    \text{System's utility}(z_{\min}) &= \TP - \FP \\
    &= u_0 + \underbrace{\int_{x\in R_M}\Pr(y = 1|X = x)dx}_{\TP~\text{from agents taking manipulation}} + \underbrace{\int_{x\in R_R}\Pr(y = 1|X = z_{min})dx}_{\TP~\text{from agents taking recourse}}\\
    & \qquad - \underbrace{\int_{x\in R_M} \left(1 - \Pr(y = 1|X = x)\right)dx}_{\FP~\text{from agents taking manipulation}} - \underbrace{\int_{x\in R_R} (1 - \Pr(y = 1|X = z_{min}))dx}_{\FP~\text{from agents taking recourse}}\\
    &= u_0 +  \int_{x\in R_M} \left(2\cdot \Pr(y = 1|X = x) - 1\right)dx + \int_{x\in R_R}\left(2\Pr(y = 1|X = z_{\min}) - 1\right)dx\\
    &= u_0 + \int_{x\in R_M} \left(2\cdot \Pr(y = 1|X = x) - 1\right)dx + \left(2\Pr(y = 1|X = z_{\min}) - 1\right) \int_{x\in R_R}dx
\end{align*}

    where $z_{\min} = \argmin_{z\in \ZZ} z$ is the cheapest recourse actions.

    Useful facts:
    \begin{enumerate}
        \item Suppose the classifier is a threshold classifier: $f = \mathbbm{I}[x\geq \theta]$, we can further characterize the  $\X^{(0)} = \{x\in \X: x\leq \theta\}$. 
        \item the minimum value of $z_{\min}$ is $\theta$ (the decision boundary).
        \item Since $\Pr[y = 1| X = x]$ is monotonic in x, $\forall x\in R_M, \Pr[y = 1| X = x] \leq \Pr[y = 1| X = \z_{min}]$
    \end{enumerate}

When we change the subsidy level $\alpha$, the two regions change as: 
    \begin{align*}
        \X^{(\alpha)}_M = \{x\in \XX^{(0)}: c_M(x, z_{\min}) < \min(1, c_R(x, z_{\min}; \alpha))\}\\
        R_R^{(\alpha)}= \{x\in \XX^{(0)}: c_R(x, z_{\min};\alpha) < \min(1, c_M(x, z_{\min};\alpha))\}
    \end{align*}
where $c_R(x, x';\alpha) = (1 - \alpha) \cdot c_R(x, x')$. As $\alpha$ becomes larger, we should expect $|\X_{R}^{(\alpha)}|$ to be larger and $|\X_{M}^{(\alpha)}|$ to be smaller. 


When $c_R(x, x')$ and $c_M(x, x')$ are both monotonic in $\|x-x'\|$ and only cross once. wlog, assume 
\begin{align*}
    c_M(x, x') = \|x - x'\|, c_R(x, x';\alpha) = \alpha \cdot w_R\cdot \|x - x'\| + b  ~~( 0<w_R\leq 1, b<1 \text{to guarantee they only cross once})
\end{align*}

we can further characterize the two regions:
\begin{align*}
    \X^{(a)}_M = \{x: x\in [z_{\min} - \sqrt{\frac{b}{1 - \alpha\cdot w_R}}, \theta]\}, \X^{(a)}_R = \{x: x\in [z_{\min} - \sqrt{\frac{1 - b}{\alpha \cdot w_R}}, z_{\min} - \sqrt{\frac{b}{1 - \alpha\cdot w_R}}]\}
\end{align*}
which gives us the size for the two regions as:
\begin{align*}
     \left|\X^{(a)}_M \right|= \theta - z_{\min} + \sqrt{\frac{b}{1 - \alpha\cdot w_R}}, ~~~
     \left|\X^{(a)}_R \right|=  \sqrt{\frac{1 - b}{\alpha \cdot w_R}} - \sqrt{\frac{b}{1 - \alpha\cdot w_R}}  
\end{align*}
For $\alpha\in [0,1]$, the rate in which the size of $\X^{(a)}_M$ and $\X^{(a)}_R$ changes as a function of the subsidy level $\alpha$ can be expressed as:   
\begin{align*}
   \frac{\partial |\X^{(a)}_M|}{\partial \alpha} &= \frac{1}{2} \cdot b^{1 / 2} \cdot w \cdot(1-a w)^{-3 / 2} , \\
   \frac{\partial |\X^{(a)}_R|}{\partial \alpha} &=
   -\frac{1}{2}\sqrt{\frac{1-b}{w}}\cdot a^{-3/2} - \frac{1}{2} \cdot b^{1 / 2} \cdot w \cdot(1-a w)^{-3 / 2} 
\end{align*}
we can see the increase rate in the size of $R_R^{(\alpha)}$ is higher than the decrease rate in the size of $R_M^{(\alpha)}$. This, together with the fact that useful fact (3), tell us that the system's utility will be a monotonically increasing function in subsidy level $\alpha$.
\end{proof*}

\subsection{Proof for Theorem \ref{thm:sub-recourse-rate-diff}}


\begin{proof}
\label{proof:sub-recourse-rate-diff}
    Recall from the proof for the recourse rate with subsidy, for a particular reveal set $\ZZ$ and a given set of negatively classified feature set $\XX_{-}$, we have:
    \begin{align*}
        \rec(\ZZ, \XX_{-}; \alpha) = \frac{\underset{\x\in\XX_{-}}{\sum} \indicator
        \bigg[\alpha\geq {1 - \frac{\min\big(1,~\underset{\z''\in\ZZ}{\min}~c_M(\x, \z'')\big)}{\underset{\z'\in\ZZ}{\min}~  c_R(\x, \z')}}
        \bigg]}{|\XX_{-}|}
    \end{align*}
To ease the notation, let's define $\gamma(x) = \frac{\underset{\x\in\XX_{-}}{\sum} 
\indicator 
        \bigg[\alpha\geq {1 - \frac{\min\big(1,~\underset{\z''\in\ZZ}{\min}~c_M(\x, \z'')\big)}{\underset{\z'\in\ZZ}{\min}~  c_R(\x, \z')}}
        \bigg]}{|\XX_{-}|}$. 
Plug the expression into the definition for the disparity in recourse ratio for two groups $g_0, g_1$, we have:
\begin{align*}
    \textit{Diff}^{(\rec)}(\ZZ, \XX_{-}^{(g_0)}, \XX_{-}^{(g_1)})
    =& \left|\rec(\ZZ, \XX_{-}^{(g_1)}, \alpha) - \rec(\ZZ, \XX_{-}^{(g_0)}, \alpha)\right|   \\
    =& \left|\frac{\underset{\x\in \XX_{-}^{(g_1)}}{\sum} \mathbbm{1}
        \bigg[\alpha\geq {1 - \gamma(x)}
        \bigg]}{|\XX_{-}^{(g_1)}|}
        -\frac{\underset{\x\in\XX_{-}^{(g_0)}}{\sum} \mathbbm{1}
        \bigg[\alpha\geq {1 - \gamma(x)}
        \bigg]}{|\XX_{-}^{(g_0)}|}
        \right|
\end{align*}

when the size of the two groups are similar, namely when $|\XX_{-}^{(g_0)}|\approx |\XX_{-}^{(g_1)}| $, we can roughly approximate the recourse difference by:
\begin{align*}
     \textit{Diff}^{(\rec)}(\ZZ, \XX_{-}^{(g_0)}, \XX_{-}^{(g_1)}, \alpha)
    \approxeq & \left|\underset{\x\in \XX_{-}^{(g_1)}}{\sum} \mathbbm{1}
        \bigg[\alpha\geq {1 - \gamma(x)}
        \bigg]
        -\underset{\x\in\XX_{-}^{(g_0)}}{\sum} \mathbbm{1}
        \bigg[\alpha\geq {1 - \gamma(x)}
        \bigg]\right|\\
\end{align*}
We make the following observation:
\begin{itemize}
    \item When $\alpha = 0$: it corresponds to the situation where no subsidy is provided. This is the original disparity $ \textit{Diff}^{(\rec)}(\ZZ, \XX_{-}^{(g_0)}, \XX_{-}^{(g_1)})$.  
    \item When $\alpha = \alpha_{\max} = 1$, it corresponds to when the cost of recourse is $0$, in this case, everyone takes recourse, which means the recourse difference is zero. 
    Since $1 - \gamma(x)\leq 1 = \alpha_{\max}$ is also an upper bound on the value $1 = \gamma(x)$ for all $x\in \XX_{-}$.
\end{itemize}
For each group $g_0$ and $g_1$, if we rank $x$ by their $ 1- \gamma(x)$ value, then as we move $\alpha$ from $0$ to $1$, all the points that are to the left of the $\alpha$ will be counted towards $\mathbbm{1}[\alpha \geq 1 - \gamma(x)]$. Thus the disparity will depend on the distribution of $1 - \gamma(x)$, which will mainly depend on the distribution of $x$, as well as the cost functions $c_R$ and $c_M$. 
However, we are guaranteed to at least find an $1<\alpha^*< 1$, such that after $\alpha > \alpha^*$, there is only one $x\in \XX^{(g_0)}_{-}$ such that $\alpha \geq 1 - \gamma(x)$ is true. In this case, increasing $\alpha$ will only leads to decreasing in the disparity.
\end{proof}

\subsection{Proof for \cref{thm:sub-social-cost-diff}}

Recall the statement of \cref{thm:sub-social-cost-diff}:
\begin{theorem*}[Subsidy Influence on Social Cost Disparity]
    With subsidy $\alpha$, the disparity in social cost for two group $g_0, g_1$ becomes:  
    \begin{align*}
         \textit{Diff}^{(\text{cost})}(\ZZ, \XX^{(g_0)}, \XX^{(g_1)};\alpha):= \left|\text{cost}(\ZZ, \XX_{-}^{(g_1)};\alpha) - \text{cost}(\ZZ, \XX_{-}^{(g_0)};\alpha)\right|   
    \end{align*}
Given a revealed set $\ZZ$, 
the social cost difference monotonically decreases in subsidies. 
\end{theorem*}

\begin{proof}
\label{proof:thm-sub-social-cost}
Recall the definition of social cost difference:
\begin{align*}
    \textit{Diff}^{(\text{cost})}(S, \XX_{-}^{(g_0)}, \XX_{-}^{(g_1)}):= \left|\text{cost}(S, \XX_{-}^{(g_1)}) - \text{cost}(S, \XX_{-}^{(g_0)})\right|
\end{align*}
Again, consider a 1-dimensional setting, where the system uses a linear threshold classifier $f(x) = \mathbbm{1}[x\geq \tau]$. In this case, the optimal recourse action for any agent is always the minimum recourse actions that has been revealed so far, namely $z_{\min} = \min_{z\in \ZZ} z$. Recall from the proof for social cost with subsidy, we have for a particular set $\X$:

 \begin{align*}
      \text{cost}(\ZZ, \XX, \alpha) 
      &= (1 - \alpha)\cdot |\XX|\cdot w_R\cdot \left(\min_{z\in \ZZ} z - \tau \right)
    \end{align*}   
Plug it back to the definition of social cost difference at a certain subsidy level, we have:
\begin{align*}
    \textit{Diff}^{(\text{cost})}(\ZZ, \XX^{(g_0)}, \XX^{(g_1)};\alpha)&= \left|\text{cost}(\ZZ, \XX^{(g_1)}) - \text{cost}(\ZZ, \XX^{(g_0)})\right|\\
    &= \left|(1 - \alpha) \cdot |\XX_{-}^{(g_0)}| \cdot (\min_{z\in \ZZ} - \tau) - (1 - \alpha) \cdot |\XX_{-}^{(g_1)}| \cdot (\min_{z\in \ZZ} - \tau)\right|\\
    &= \left|(1 - \alpha) \cdot (|\XX_{-}^{(g_0)}| - |\XX_{-}^{(g_1)}|)\cdot (\min_{z\in \ZZ} - \tau)\right|
\end{align*}

which is monotonically decreasing as $\alpha$ increases.
    
\end{proof}

\section{Additional Experimental Results}
\label{sec:additional-experimental-results}

\textbf{Additional Experimental Setup}
To optimize the system's utility and select the optimal set of features to reveal, we use the local search-based method provided in \cite{orso2015submodular}.

\textbf{Experiments Compute Resources}
All the experiments were run on a MacBook Pro with Apple M1 chip and 8GB memory. To finish 100 runs on the adult and law datasets for 5 different subsidies, it took roughly 3 hours; the German credit dataset will take roughly 5-6 hours. 

\subsection{Additional result on recourse rate difference between groups as a function of different values of subsidies}

In Figure \ref{fig:lr_rec_rate_diff_sub}, we see the recourse rate difference between groups as a function of different values of subsidies. 
This figure serves to outline the parabolic nature relationship between subsidies and recourse rate difference.
As mentioned previously, only those with already low recourse costs can benefit from subsidies for smaller subsides.
Thus we see that smaller subsidies can initially result in greater disparity between agents, however, as subsides increases, they eventually decrease disparity to rates which are lower than the disparity without subsidies ($sub=0$).
Thus when deciding the amount of subsidies to choose, it is important for systems to be aware of the potential negative impacts (larger disparities between groups) that can result from smaller subsidies.

\label{sec:rec_ratio_diff_sub}
\begin{figure}[tb!]
        \includegraphics[width=\textwidth]{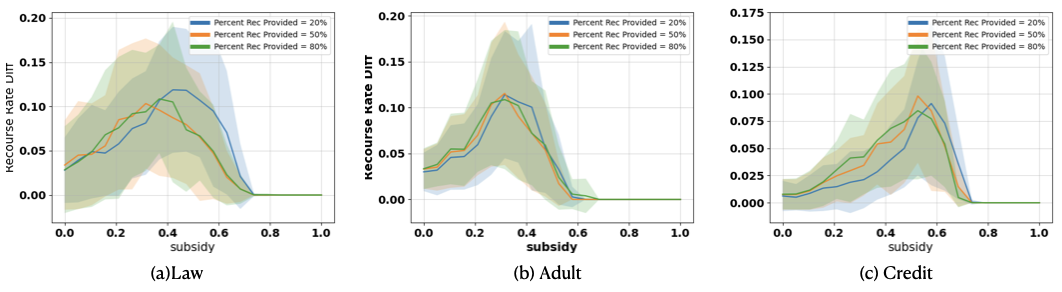}
     \caption{Recourse rate difference as a function of subsidy with 95\% confidence intervals. Each line corresponds to a different percentage of the population with provided recourse actions.}
     \label{fig:lr_rec_rate_diff_sub}
\end{figure}

\subsection{Additional Results Using Gradient Boosting Classifier}
In this section, we present further empirical findings obtained by employing a Gradient Boosting Decision Tree as the training method. Overall, we observed similar behavior compared with training and logistic regression. 

\begin{figure}[tb!]
    \centering
    \begin{subfigure}[b]{0.3\textwidth}
        \centering
        \includegraphics[width=\textwidth]{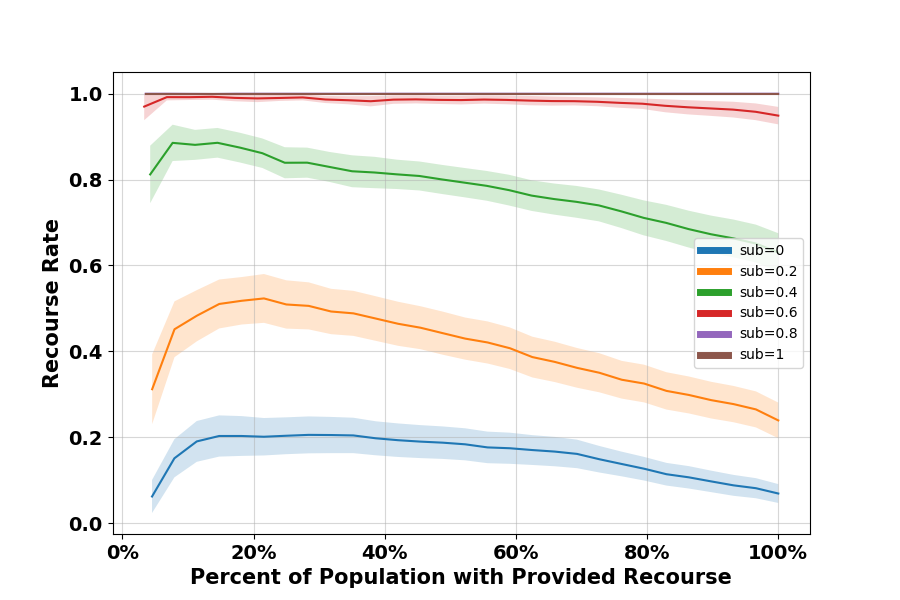}
        \caption{Law}
        \label{fig:law_rec_ratio}
    \end{subfigure}
    \hspace{-10pt}
    \begin{subfigure}[b]{0.3\textwidth}
        \centering
        \includegraphics[width=\textwidth]{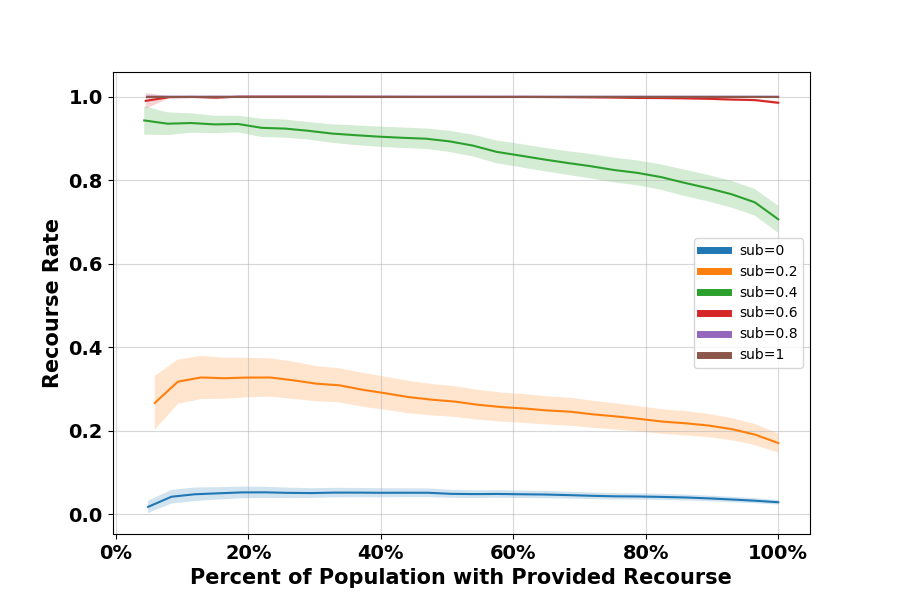}
        \caption{Adult}
        \label{fig:adult_rec_ratio}
    \end{subfigure}
    \hspace{-10pt}
    \begin{subfigure}[b]{0.3\textwidth}
        \centering
        \includegraphics[width=\textwidth]{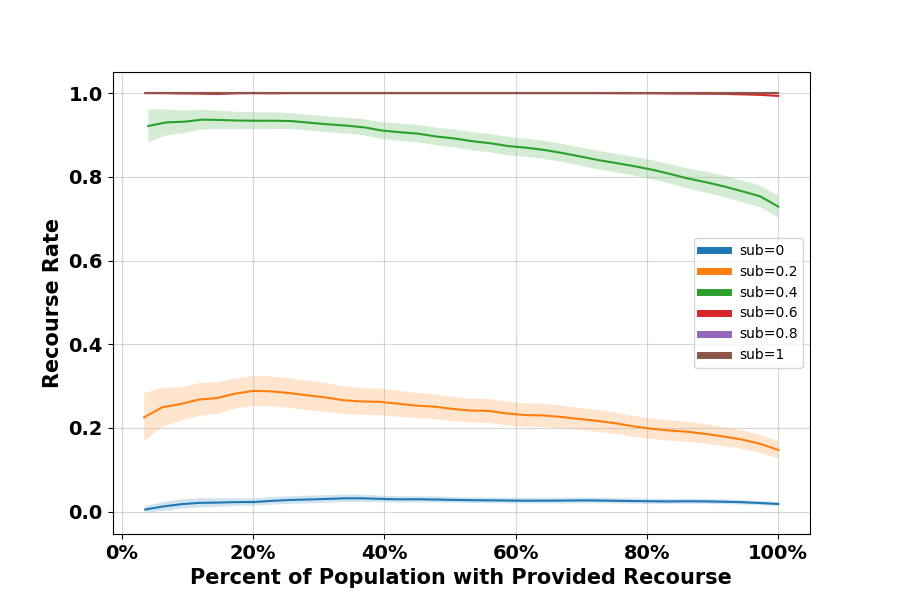}
        \caption{Credit}
        \label{fig:german_rec_ratio_gbc}
    \end{subfigure}
    \caption{Fraction of the population performing recourse, with 95\% confidence intervals. Each line corresponds to a different subsidy ratio ``subs", i.e., the cost reduction applied to recourse. }
    \label{fig:_gbc-recourse-rate}
\end{figure}


\begin{figure}[tb]
    \centering
    \begin{subfigure}[b]{0.3\textwidth}
        \centering
        \includegraphics[width=\textwidth]{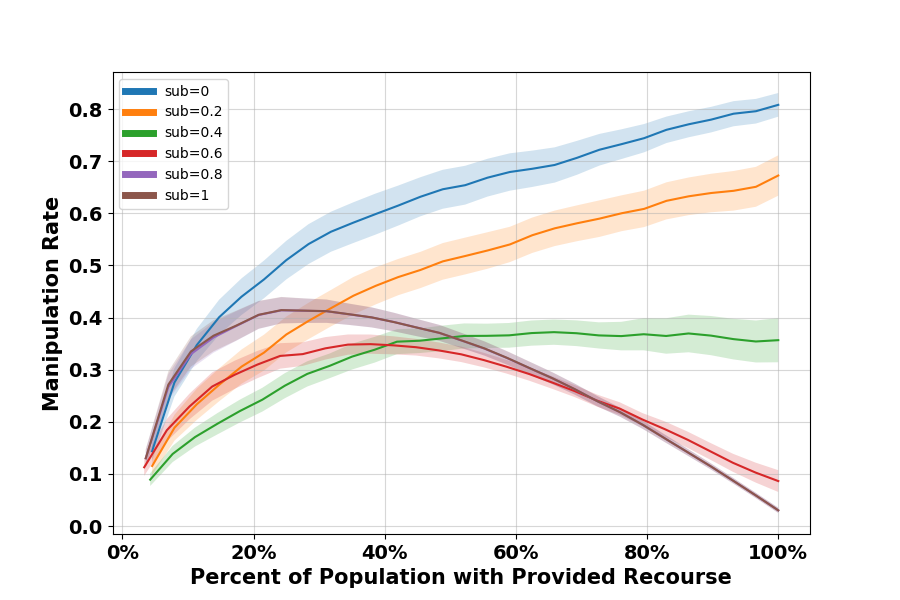}
        \caption{Law}
        \label{fig:law_man_ratio_gbc}
    \end{subfigure}
    \hspace{-10pt}
    \begin{subfigure}[b]{0.3\textwidth}
        \centering
        \includegraphics[width=\textwidth]{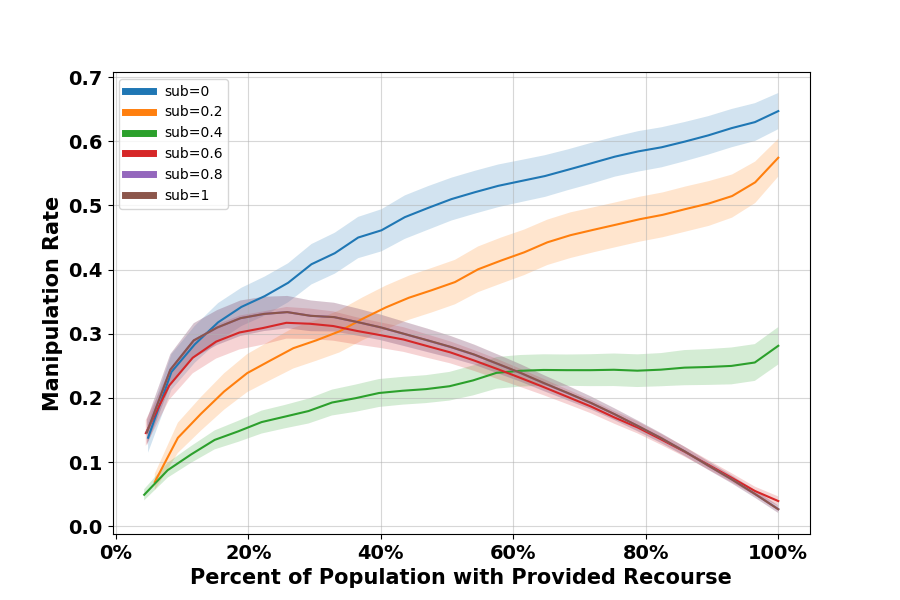}
        \caption{Adult}
        \label{fig:adult_man_ratio_gbc}
    \end{subfigure}
    \hspace{-10pt}
    \begin{subfigure}[b]{0.3\textwidth}
        \centering
        \includegraphics[width=\textwidth]{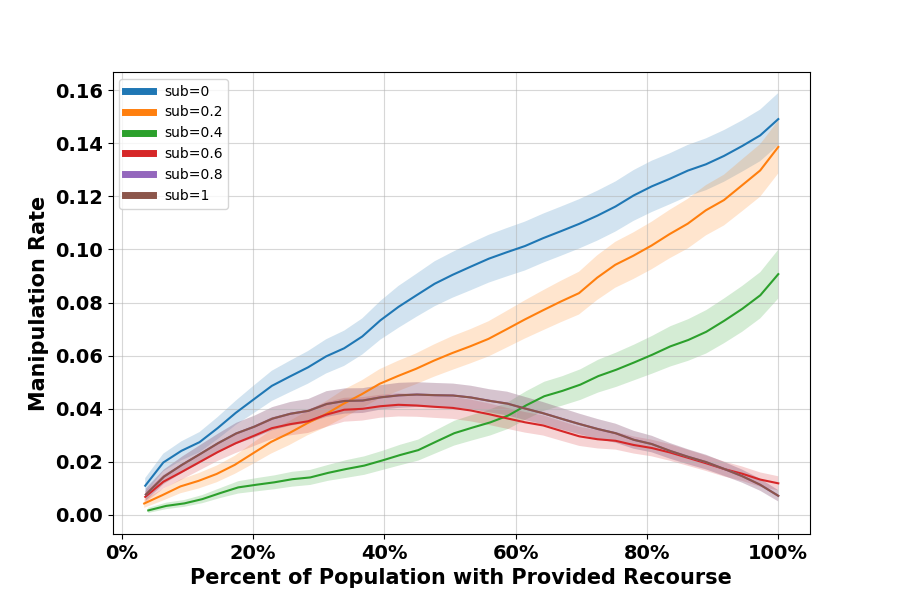}
        \caption{Credit}
        \label{fig:german_man_ratio_gbc}
    \end{subfigure}
    \caption{Fraction of the population performing manipulation, with 95\% confidence intervals. Each line corresponds to a different subsidy ratio ``subs", i.e., the cost reduction applied to recourse. }
    \label{fig:gbc-manipulation-rate}
\end{figure}

\begin{figure}[bt]
    \centering
    \begin{subfigure}[b]{0.3\textwidth}
        \centering
        \includegraphics[width=\textwidth]{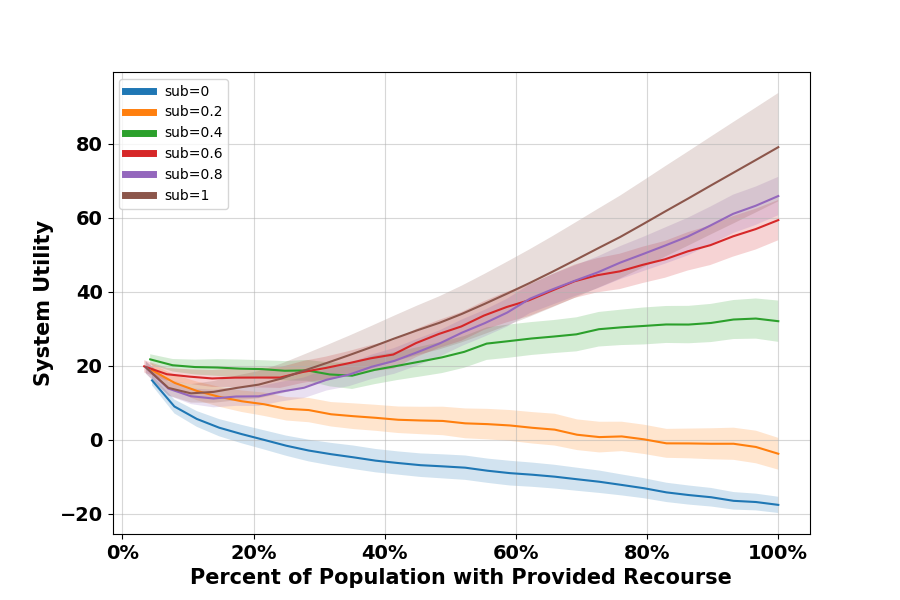}
        \caption{Law}
        \label{fig:law_utility_gbc}
    \end{subfigure}
    \hspace{-10pt}
    \begin{subfigure}[b]{0.3\textwidth}
        \centering
        \includegraphics[width=\textwidth]{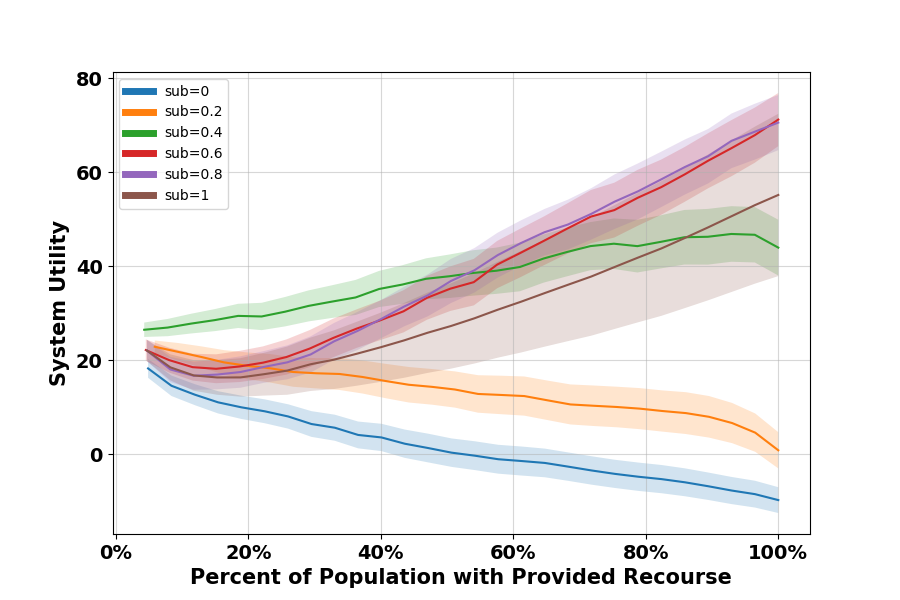}
        \caption{Adult}
        \label{fig:adult_utility_gbc}
    \end{subfigure}
    \hspace{-10pt}
    \begin{subfigure}[b]{0.3\textwidth}
        \centering
        \includegraphics[width=\textwidth]{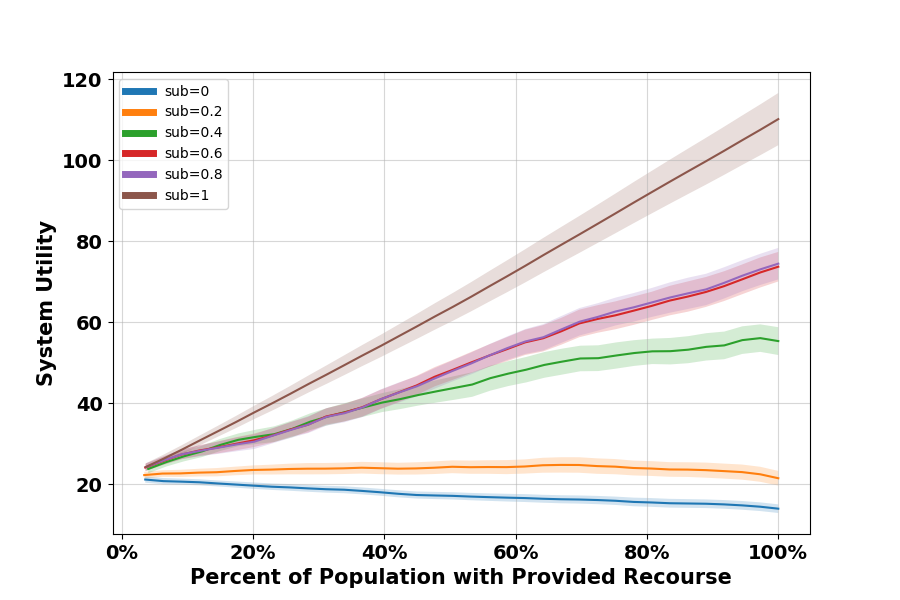}
        \caption{Credit}
        \label{fig:german_utility_gbc}
    \end{subfigure}
    \caption{The system's utility as a function of the population percentage with provided recourse, with 95\% confidence intervals. Each line corresponds to a different subsidy ratio ``subs", i.e., the cost reduction applied to recourse. }
    \label{fig:gbc-utility}
\end{figure}

\begin{figure}[bt]
    \centering
    \begin{subfigure}[b]{0.3\textwidth}
        \centering
        \includegraphics[width=\textwidth]{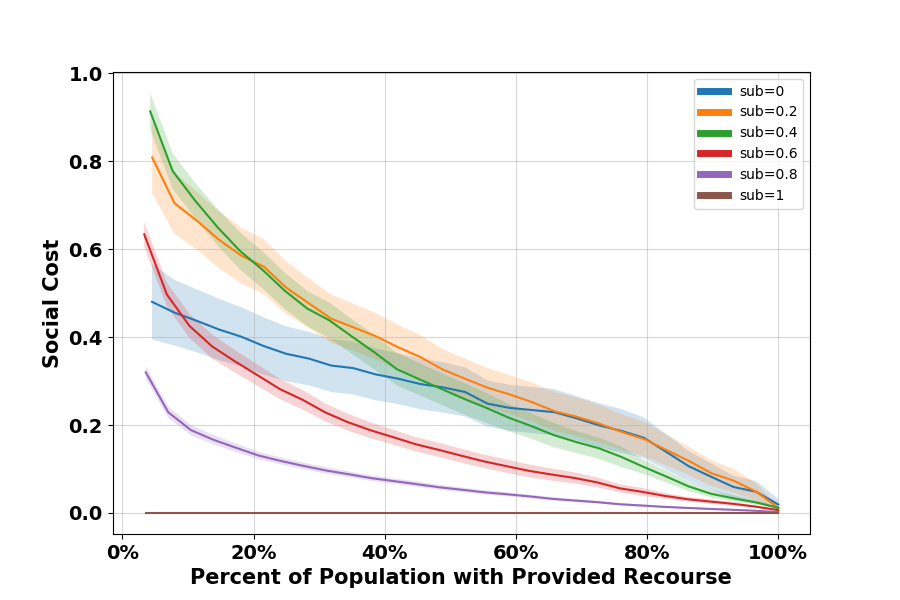}
        \caption{Law}
        \label{fig:law_social_cost_gbc}
    \end{subfigure}
    \hspace{-10pt}
    \begin{subfigure}[b]{0.3\textwidth}
        \centering
        \includegraphics[width=\textwidth]{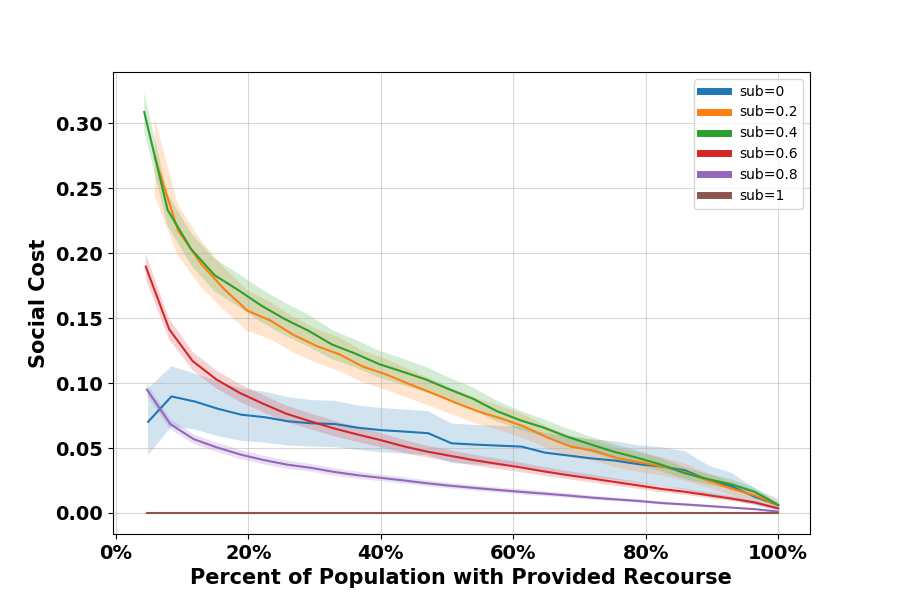}
        \caption{Adult}
        \label{fig:adult_social_cost}
    \end{subfigure}
    \hspace{-10pt}
    \begin{subfigure}[b]{0.3\textwidth}
        \centering
        \includegraphics[width=\textwidth]{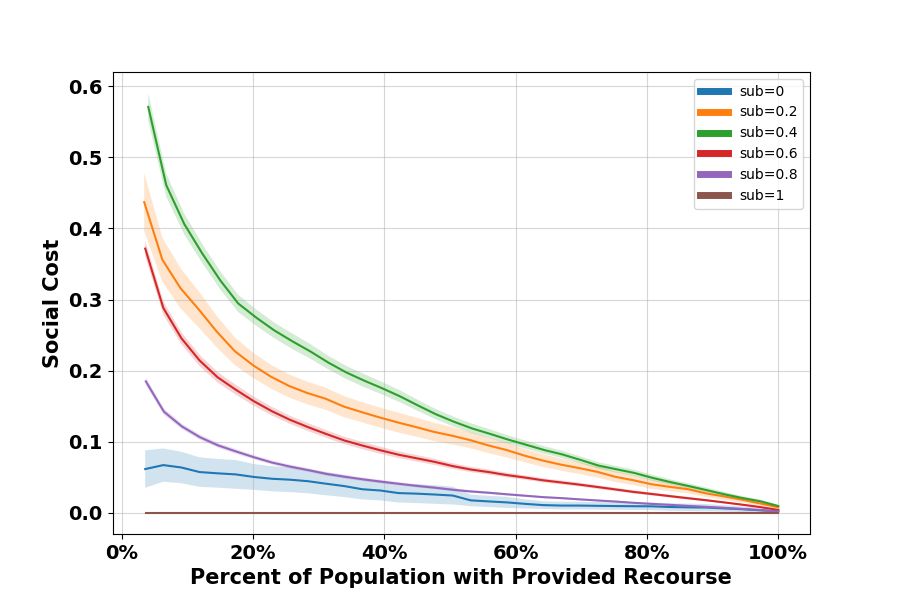}
        \caption{Credit}
        \label{fig:german_social_cost_gbc}
    \end{subfigure}
    \caption{The social cost as a function of the population percentage with provided recourse, with 95\% confidence intervals. Each line corresponds to a different subsidy ratio ``subs", i.e., the cost reduction applied to recourse. }
    \label{fig:gbc-social-cost}
\end{figure}

\begin{figure}[bt]
    \centering
    \begin{subfigure}[b]{0.3\textwidth}
        \centering
        \includegraphics[width=\textwidth]{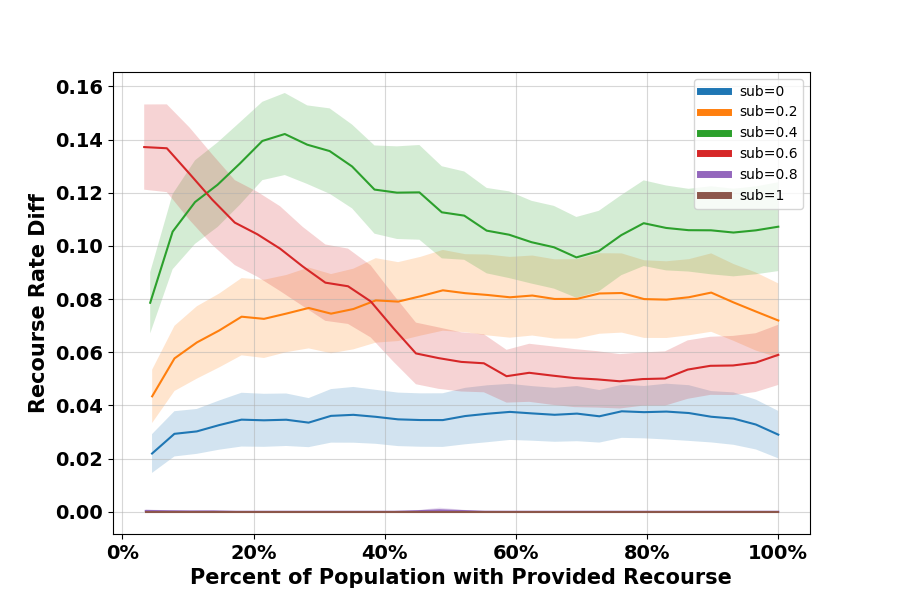}
        \caption{Law}
        \label{fig:law_rec_rate_diff_gbc}
    \end{subfigure}
    \hspace{-10pt}
    \begin{subfigure}[b]{0.3\textwidth}
        \centering
        \includegraphics[width=\textwidth]{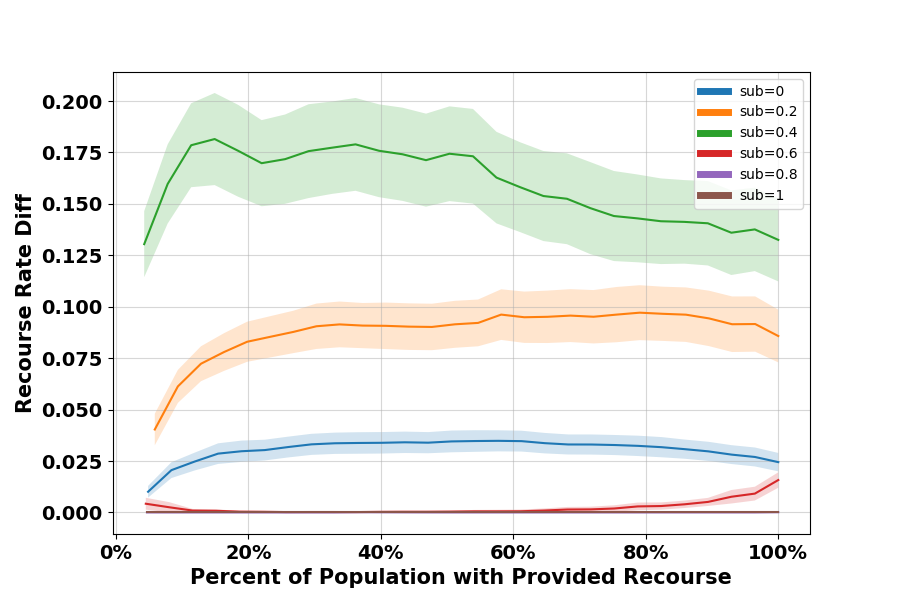}
        \caption{Adult}
        \label{fig:adult_rec_rate_diff_gbc}
    \end{subfigure}
    \hspace{-10pt}
    \begin{subfigure}[b]{0.3\textwidth}
        \centering
        \includegraphics[width=\textwidth]{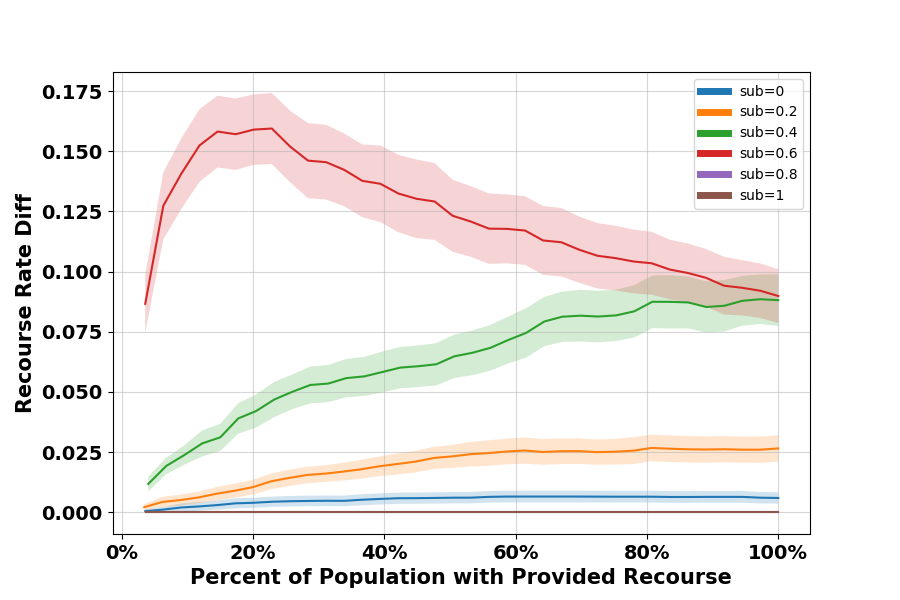}
        \caption{Credit}
        \label{fig:german_rec_rate_diff_gbc}
    \end{subfigure}
    \caption{Difference in recourse rate between different sensitive attribute groups with 95\% confidence intervals. Each line corresponds to a different subsidy ratio ``subs", i.e., the cost reduction applied to recourse. }
    \label{fig:gbc-rec-rate_diff}
\end{figure}

\begin{figure}[bt]
    \centering
    \begin{subfigure}[b]{0.3\textwidth}
        \centering
        \includegraphics[width=\textwidth]{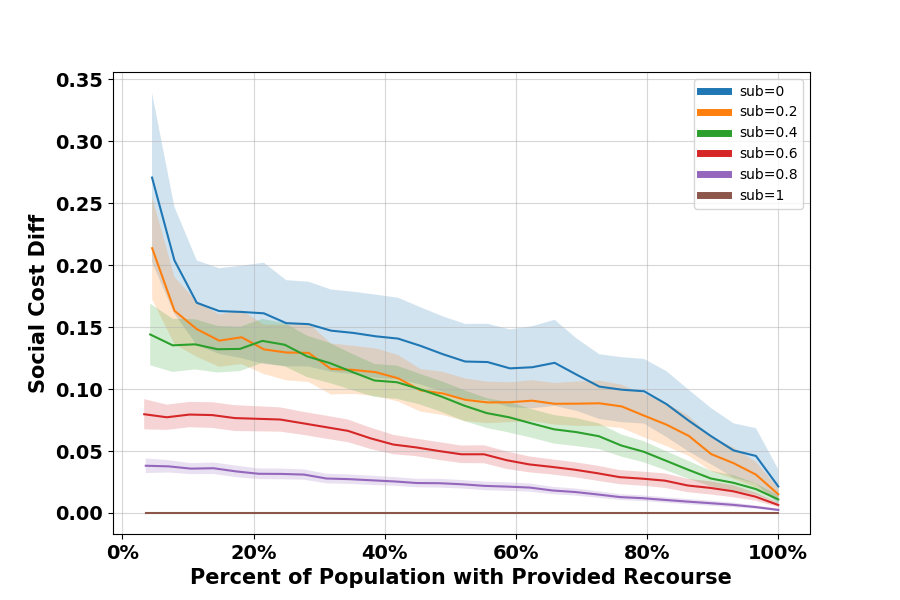}
        \caption{Law}
        \label{fig:law_social_cost_diff_gbc}
    \end{subfigure}
    \hspace{-10pt}
    \begin{subfigure}[b]{0.3\textwidth}
        \centering
        \includegraphics[width=\textwidth]{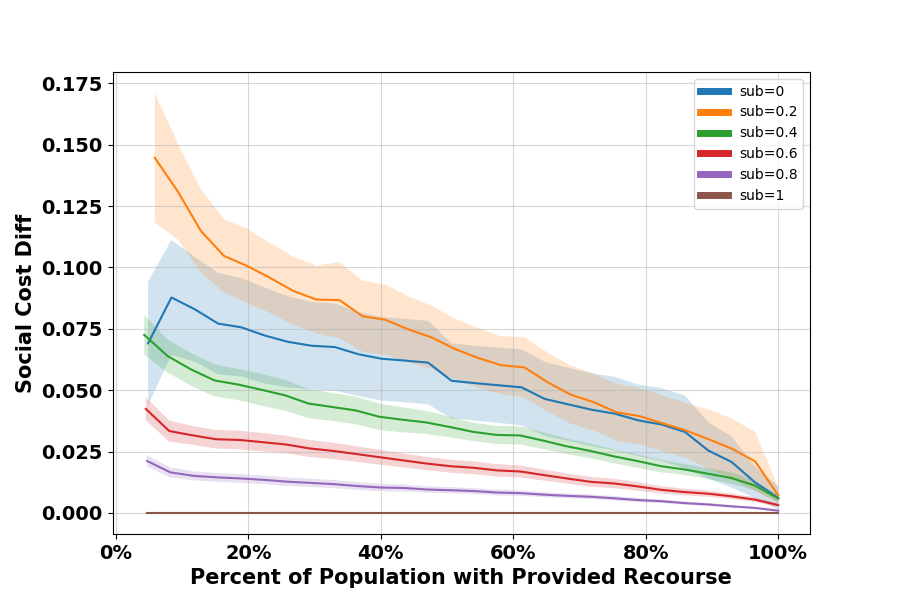}
        \caption{Adult}
        \label{fig:adult_social_cost_diff_gbc}
    \end{subfigure}
    \hspace{-10pt}
    \begin{subfigure}[b]{0.3\textwidth}
        \centering
        \includegraphics[width=\textwidth]{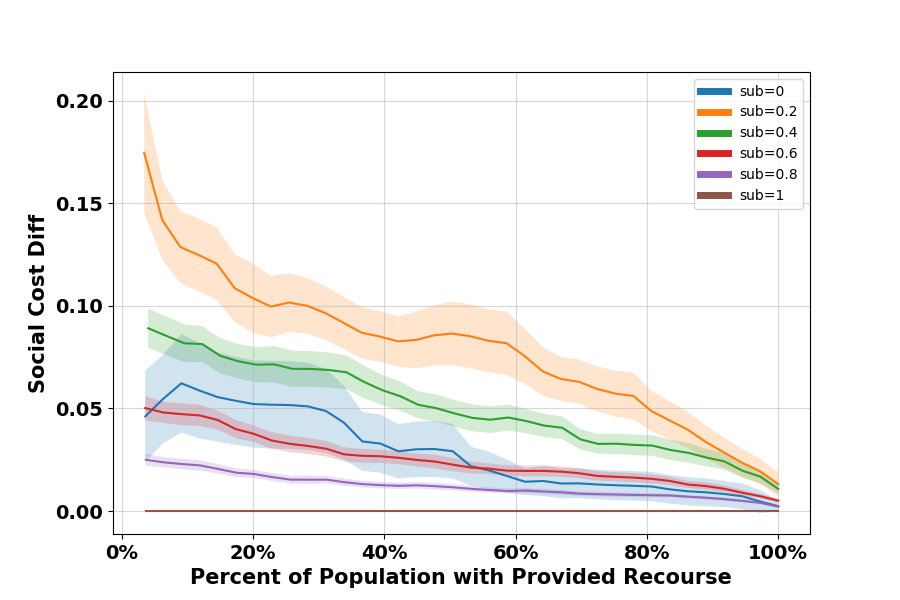}
        \caption{Credit}
        \label{fig:german_social_cost_diff_gbc}
    \end{subfigure}
    \caption{Difference in social cost between different sensitive attribute groups with 95\% confidence intervals. Each line corresponds to a different subsidy ratio ``subs", i.e., the cost reduction applied to recourse. }
    \label{fig:gbc-social-cost-diff}
\end{figure}

\begin{figure}[h!]
    \centering
    \begin{subfigure}[b]{0.3\textwidth}
        \centering
        \includegraphics[width=\textwidth]{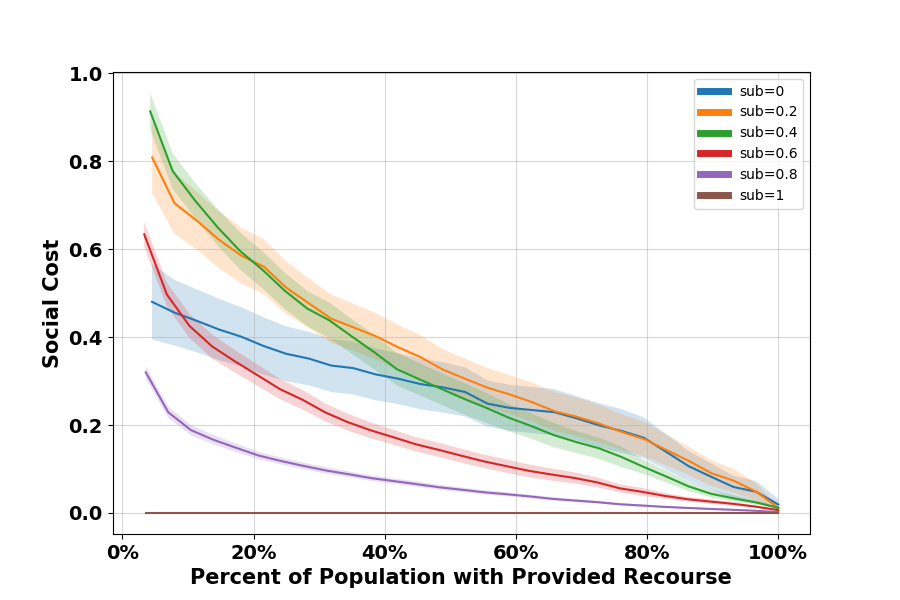}
        \caption{Law}
        \label{fig:law_social_cost_gbc}
    \end{subfigure}
    \hspace{-10pt}
    \begin{subfigure}[b]{0.3\textwidth}
        \centering
        \includegraphics[width=\textwidth]{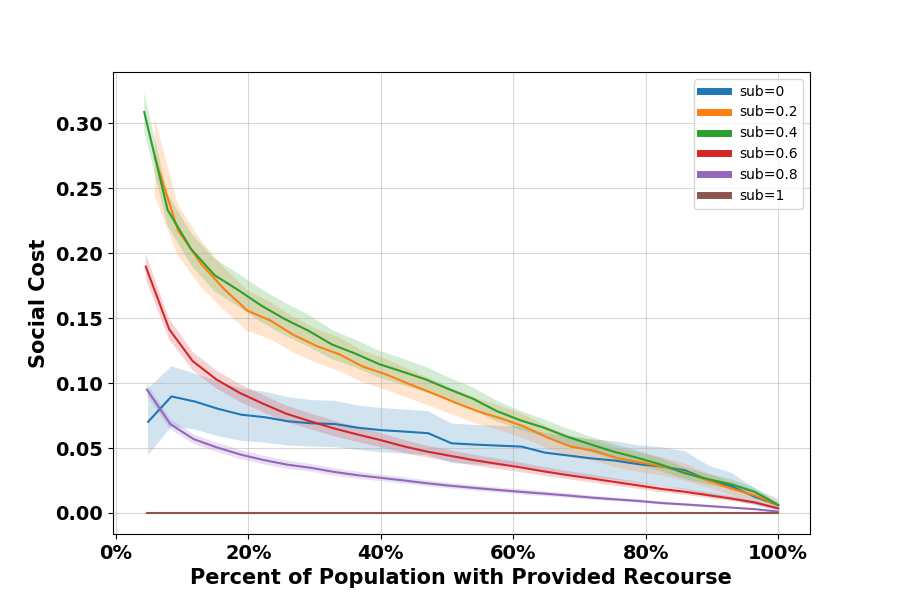}
        \caption{Adult}
        \label{fig:adult_social_cost_gbc}
    \end{subfigure}
    \hspace{-10pt}
    \begin{subfigure}[b]{0.3\textwidth}
        \centering
        \includegraphics[width=\textwidth]{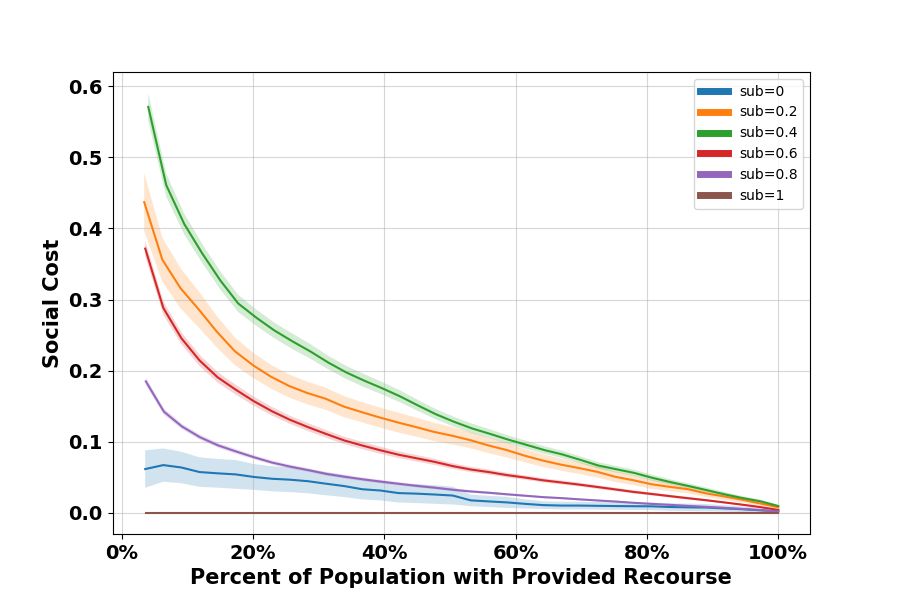}
        \caption{Credit}
        \label{fig:german_social_cost_gbc}
    \end{subfigure}
    \caption{The social cost as a function of the population percentage with provided recourse, with 95\% confidence intervals. Each line corresponds to a different subsidy ratio ``subs", i.e., the cost reduction applied to recourse. }
    \label{fig:gcb-lr-social-cost}
\end{figure}

\newpage
\section{Boarder Impact}

\textbf{Boarder Impact:} By shedding light on the complex dynamics of recourse provision in automated systems, our paper challenges existing assumptions and reveals significant implications for both individuals and society as a whole. The identification of the natural tension between providing recourse and system exploitation highlights the delicate balance that must be maintained in algorithmic decision-making. This insight has profound consequences for fairness and equity, as strategic recourse withholding disproportionately affects vulnerable groups. Furthermore, the proposed framework offers a novel approach to analyzing the interplay of transparency, recourse, and manipulation, providing a valuable tool for future research in algorithmic fairness and accountability. The findings underscore the urgent need for policy interventions, such as recourse subsidies, to mitigate the adverse effects of system behavior on marginalized populations. Ultimately, this paper not only advances our theoretical understanding of algorithmic decision-making but also offers practical solutions to address the systemic biases inherent in automated systems.

\end{document}